\documentclass[a4paper,11pt]{amsart}

\usepackage{amsmath,amsthm}
\usepackage{enumerate}
\usepackage{xcolor}
\usepackage{xspace}
\usepackage{amsopn}
\usepackage{mathrsfs}
\usepackage{xparse}
\usepackage{upgreek}
\usepackage[utf8]{inputenc}
\usepackage{t4phonet}
\usepackage{enumitem}
\usepackage[]{algorithm2e}
\usepackage{hyperref}
\usepackage[foot]{amsaddr}

\usepackage{tikz}
\usetikzlibrary{shapes, fit, decorations.pathreplacing}

\newtheorem{theorem}{Theorem}[section]
\newtheorem{corollary}[theorem]{Corollary}
\newtheorem{lemma}[theorem]{Lemma}
\newtheorem{proposition}[theorem]{Proposition}
\newtheorem*{11ORT}{1-1 Operator Recursion Theorem}

\theoremstyle{definition}
\newtheorem{definition}[theorem]{Definition}

\theoremstyle{remark}

\theoremstyle{theorem}

\theoremstyle{remark}
\newtheorem*{examplenonum}{Example}

\newcounter{caseCounter}[theorem]

\newcommand*{\ORT}{\mathrm{ORT}\xspace}

\newcommand*{\Txt}{\mathbf{Txt}}
\newcommand*{\Inf}{\mathbf{Inf}}

\newcommand*{\Fn}{\mathbf{Fn}}

\newcommand*{\Sd}{\mathbf{Sd}}

\newcommand*{\Meth}{\Sigma}

\newcommand*{\Ex}{\mathbf{Ex}}

\newcommand*{\Lim}{\mathbf{Ex}}
\newcommand*{\Bc}{\mathbf{Bc}}
\newcommand*{\Comp}{\mathrm{Cons}}

\newcommand*{\Cons}{\mathbf{Cons}}
\newcommand*{\Conv}{\mathbf{Conv}}
\newcommand*{\Caut}{\mathbf{Caut}}

\newcommand*{\NU}{\mathbf{NU}}
\newcommand*{\SNU}{\mathbf{SNU}}
\newcommand*{\Dec}{\mathbf{Dec}}
\newcommand*{\SDec}{\mathbf{SDec}}
\newcommand*{\SynDec}{\mathbf{SynDec}}
\newcommand*{\WMon}{\mathbf{WMon}}
\newcommand*{\Mon}{\mathbf{Mon}}
\newcommand*{\SMon}{\mathbf{SMon}}

\newcommand*{\N}{\mathbb{N}}
\newcommand*{\natnum}{\mathbb{N}}
\newcommand*{\Nttwo}{\N\!\times\!\{0,1\}}

\newcommand*{\Pow}{\mathrm{Pow}}

\newcommand*{\CalL}{\mathcal{L}} 

\newcommand*{\CalI}{\mathbf{Inf}}

\newcommand*{\CalF}{\mathcal{F}}
\newcommand*{\Simu}{\mathfrak{S}}
\newcommand*{\simu}{\mathfrak{s}}
\renewcommand*{\o}{\omega}
\newcommand*{\Seq}[1]{{\text{$#1^{< \o}$}}}
\newcommand{\ISeq}[1]{\text{$#1^{\o}$}}
\newcommand{\IfiSeq}[1]{\text{$#1^{\leq \o}$}}
\newcommand*{\totalFn}{\mathfrak{R}}
\newcommand*{\partialFn}{\mathfrak{P}}
\newcommand*{\partFn}{\mathfrak{P}}
\newcommand*{\totalCp}{\mathcal{R}}
\newcommand*{\partialCp}{\mathcal{P}}
\newcommand*{\CInd}{C\mathbf{Ind}}

\newcommand{\dom}{\mathrm{dom}}

\newcommand{\ran}{\mathrm{ran}}

\newcommand{\ps}{\mathrm{pos}}
\renewcommand{\ng}{\mathrm{neg}}
\newcommand{\concat}{^\smallfrown}
\newcommand{\pr}{\mathrm{pr}}

\newcommand{\ind}{\mathrm{ind}}
\newcommand{\pad}{\mathrm{pad}}

\newcommand*{\inseg}{\sqsubseteq}
\newcommand*{\prinseg}{\sqsubsetneq}

\newcommand*{\can}{\mathrm{can}}

\newcommand*{\decode}{\mathrm{decode}}
\newcommand*{\classifier}{\chi}
\newcommand*{\Diff}{\mathrm{Diff}}
\newcommand*{\set}[2]{\{ #1 \mid #2 \}}

\usepackage{amssymb,amsfonts,ifsym,mathrsfs,mathabx,mathbbol,wasysym}
\usepackage[only, trianglelefteqslant, trianglerighteqslant, ntrianglelefteqslant, varoplus]{stmaryrd}
\usepackage{calligra}

\providecommand{\ignore}[1]{}

\title[Learning from Informant]{Learning from Informant: \\
Relations between Learning Success Criteria}
\author{Martin Aschenbach \and Timo Kötzing \and Karen Seidel}
\address{Hasso Plattner Institute \\
        Prof.-Dr.-Helmert-Str. 2-3 \\
        D-14482 Potsdam}
\email{karen.seidel@hpi.de}

\begin{document}

\maketitle

\begin{abstract}
Learning from positive and negative information, so-called \emph{informant}, being one of the models for human and machine learning introduced by E.~M.~Gold is investigated.
Particularly, naturally arising questions about this learning setting, originating in results on learning from solely positive information, are answered.

\smallskip
By a carefully arranged argument learners can be assumed to only change their hypothesis in case it is inconsistent with the data (such a learning behavior is called \emph{conservative}). The deduced main theorem states the relations between the most important delayable learning success criteria, being the ones not ruined by a delayed in time hypothesis output.

\smallskip
Additionally, our investigations concerning the non-delayable requirement of consistent learning underpin the claim for \emph{delayability} being the right structural property to gain a deeper understanding concerning the nature of learning success criteria.

\smallskip
Moreover, we obtain an anomalous \emph{hierarchy} when allowing for an increasing finite number of \emph{anomalies} of the hypothesized language by the learner compared with the language to be learned.

In contrast to the vacillatory hierarchy for learning from solely positive information, we observe a \emph{duality} depending on whether infinitely many \emph{vacillations} between different (almost) correct hypotheses are still considered a successful learning behavior.

\end{abstract}


\section{Introduction}
\label{sec:intro}
Research in the area of \emph{inductive inference} aims at investigating the learning of formal languages and has connections to computability theory, complexity theory, cognitive science, machine learning, and more generally artificial intelligence.
Setting up a classification program for deciding whether a given word belongs to a certain language can be seen as a problem in supervised machine learning, where the machine experiences labeled data about the target language. The label is $1$ if the datum is contained in the language and $0$ otherwise. The machine's task is to infer some rule in order to generate words in the language of interest and thereby generalize from the training samples.
This so-called \emph{learning from informant} was introduced in \cite{Gol:j:67} and further investigated in several publications, including \cite{Blu-Blu:j:75}, \cite{Bar:c:77:aut-func-prog}, \cite{STL1}, \cite{lange1996monotonic} and \cite{pmlr-v76-holzl17a}.

\smallskip
According to \cite{Gol:j:67} the learner is modelled by a computable function, successively receiving sequences incorporating more and more data. The source of labeled data is called an \emph{informant}, which is supposed to be \emph{complete in the limit}, i.e., every word in the language must occur at least once.
Thereby, the learner possibly updates the current description of the target language (its hypothesis).
Learning is considered successful, if after some finite time the learners' hypotheses yield good enough approximations to the target language. The original and most common learning success criterion is called \emph{$\Lim$-learning} 
 and additionally requires that the learner eventually settles on exactly one correct hypothesis, which precisely captures the words in the language to be learned.
As a single language can easily be learned, the interesting question is whether there is a learner successful on all languages in a fixed collection of languages.

\begin{examplenonum}
Consider $\CalL=\{\,\N\setminus X \mid X\subseteq\N \text{ finite}\,\}$, the collection of all co-finite sets of natural numbers.
Clearly, there is a computable function $p$ mapping finite subsets $X\subseteq\N$ to $p(X)$, such that $p(X)$ encodes a program which stops if and only if the input is not in $X$.
We call $p(X)$ an \emph{index} for $\N\setminus X$.
The learner is successful if for every finite $X\subseteq\N$ it infers $p(X)$ from a possibly very large but finite number of samples labeled according to $\N\setminus X$.

Regarding this example, let us assume the first two samples are $(60,1)$ and $(2,0)$. The first datum still leaves all options with $60\notin X$. As the second datum tells us that $2 \in X$, we may make the learner guess $p(\{2\})$ until possibly more negative data is available. Thus, the collection of all co-finite sets of natural numbers is $\Lim$-learnable from informant, simply by making the learner guess the complement of all negative information obtained so far. Since after finitely many steps all elements of the finite complement of the target language have been observed, the learner will be correct from that point onward.

It is well-known that this collection of languages cannot be learned from purely positive information. Intuitively, at any time the learner cannot distinguish the whole set of natural numbers from all other co-finite sets which contain all natural numbers presented to the learner until this point.
\end{examplenonum}

Learning from solely positive information, so-called \emph{text}, has been studied extensively, including many learning success criteria and other variations. Some results are summed up in \cite{Jai-Osh-Roy-Sha:b:99:stl2} and \cite{Case2016gold}. We address the naturally arising question what difference it makes to learn from positive and negative information.

\subsection{Our Contributions}

For learning from text there are entire maps displaying the pairwise relations of different well-known learning success criteria, see \cite{kotzing2014map}, \cite{kotzing2016towards} and \cite{jain2016role}. We give an equally informative map for $\Lim$-learning from informant.

The most important requirements on the learning process when learning from informant are \emph{conservativeness} ($\Conv$), where only inconsistent hypotheses are allowed to be changed; \emph{strong decisiveness} ($\SDec$), forbidding to ever return semantically to a withdrawn hypothesis; \emph{strong monotonicity} ($\SMon$), requiring that in every step the hypothesis incorporates the former one; \emph{monotonicity} ($\Mon$), fulfilled if in every step the set of correctly inferred words incorporates the formerly correctly guessed; \emph{cautiousness} ($\Caut$), for which never a strict subset of earlier conjectures is guessed.
In \cite{lange1996monotonic} it was 
observed that requiring monotonicity is restrictive and that under the assumption of strong monotonicity even fewer collections of languages can be learned from informant.
We complete the picture by answering the following questions regarding $\Lim$-learning from informant positively:

\begin{enumerate}
\item[1.] Is every learnable collection of languages also learnable in a conservatively and strongly decisively way?
\item[2.] Are monotonic and cautious learning incomparable?
\end{enumerate}

The above mentioned observations in \cite{lange1996monotonic} follow from positively answering the second question.


\medskip
A diagram incorporating the resulting map is depicted in Figure~\ref{DiagramInfDelayableEx}.
The complete map can be found in Figure~\ref{MapInfDelayableEx}.

\medskip
Answering the first question builds on providing the two \emph{normal forms} of (1) requiring learning success only on the information presented in the canonical order and (2) assuming the learner to be defined on all input sequences. Further, a regularity property borrowed from text learning plays a crucial role in the proof.

\medskip
Requiring all of the learners guesses to be \emph{consistent} with the positive and the negative information being presented to it so far makes learning harder. Next to this we also observe that the above normal forms cannot be assumed when the learner is required to act consistently.
On the one hand, it is easier to find a learner for a collection of languages that consistently learns each of them only from the canonical presentation than finding one consistently learning them from arbitrary informant.
On the other hand finding a total learner consistently $\Lim$-learning a collection of languages is harder than finding a partial one.

We further transfer the concept of a learning success criterion to be invariant under time-delayed outputs of the hypotheses, introduced for learning from text in \cite{kotzing2016map} and generalized in \cite{KSS17}, to the setting of learning from informant.
Consistency is not \emph{delayable} since a hypothesis which is consistent now might be inconsistent later due to new data.
As this is the only requirement not being delayable, the results mentioned in the last paragraph justify the conjecture of delayability being the right property to proof more results that at once apply to all learning success criteria but consistency.

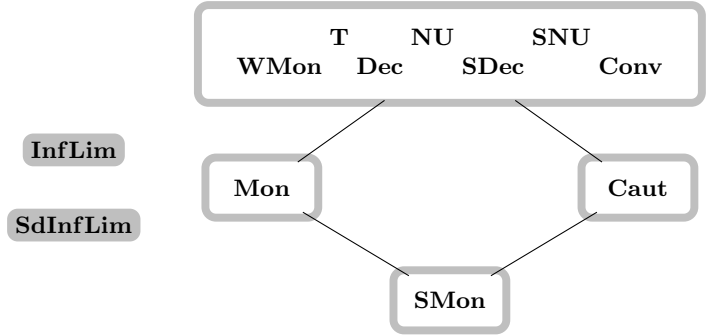
\begin{figure}[h]
\begin{center}
\begin{minipage}{12cm}
\begin{tikzpicture}[above,sloped,shorten <=2mm,, shorten >=2mm]

\begin{scope}[every node/.style={minimum size=4mm}]

\node[shape=rectangle,rounded corners,fill=lightgray] at (-5,-1) {$\Inf\Lim$};
\node[shape=rectangle,rounded corners,fill=lightgray] at (-5,-2) {$\Sd\Inf\Lim$};

\node (nothing) at (0,0)  {
$\begin{array}{ccccccc}
& \mathbf{T}  & & \NU && \SNU \\
\WMon && \Dec && \SDec && \Conv
\end{array}$
};

\node (caut) at (2.5,-1.5)		{$\Caut$};
\node (mon) at (-2.5,-1.5)		  {$\Mon$}; 
\node (smon) at (0,-3)	  {$\SMon$};

\node[draw=lightgray,line width=3pt, rounded corners, minimum width= 87mm,minimum height = 13mm, anchor=center] at (nothing) {};
\node[draw=lightgray,line width=3pt, rounded corners, minimum width= 15mm,minimum height = 8mm, anchor=center] at (caut) {};
\node[draw=lightgray,line width=3pt, rounded corners, minimum width= 15mm,minimum height = 8mm, anchor=center] at (mon) {};
\node[draw=lightgray,line width=3pt, rounded corners, minimum width= 15mm,minimum height = 8mm, anchor=center] at (smon) {};

\draw (nothing) -- (caut);
\draw (nothing) -- (mon);
\draw (caut) -- (smon);
\draw (mon) -- (smon);

\end{scope}

\end{tikzpicture}
\end{minipage}
\end{center}
\caption{Relations between delayable learning restrictions in $\Lim$-learning from informant.
Implications are represented as black lines from bottom to top.
Two learning settings are equivalent if and only if they lie in the same grey outlined zone.}
\label{DiagramInfDelayableEx}
\end{figure}

\medskip
While in \cite{lange1994characterization} variously restricted learning of collections of recursive languages with a uniform decision procedure are considered, the above mentioned results also apply to arbitrary collections of recursively enumerable sets. Further, our results are as strong as possible, meaning that negative results are stated for indexable families, if possible, and positive results for all collections of languages.

\medskip
In this spirit we add to a careful investigation on how informant and text learning relate to each other in \cite{Lan-Zeu:c:93}. We show that even for the most restrictive delayable learning success criterion when $\Lim$-learning from informant there is a collection of recursive languages learnable in this setting that is not $\Lim$-learnable from text.

\medskip
Admitting for finitely many anomalies $a$, i.e. elements of the symmetric difference of the hypothesized and the target concept, yields the anomalous hierarchy for learning from text in \cite{Cas-Lyn:c:82}. We provide an equivalence between learning collections of functions from enumerations of their graphs to learning the languages encoding their graphs from informant. This allows us to transfer the anomalous hierarchy when learning functions in \cite{Bar:j:74:two-thrm} and \cite{case1983comparison} to the setting of learning from informant and therefore obtain a hierarchy
$$[\Inf\Lim] \subsetneq \ldots \subsetneq [\Inf\Lim^a] \subsetneq [\Inf\Lim^{a+1}] \subsetneq \ldots.$$

\medskip
\cite{case1999power} observed the vacillatory hierarchy for learning from text.
There\-by in the limit a vacillation between $b$ many (almost) correct descriptions is allowed, where $b \in \N_{>0} \cup \{\infty\}$.
In contrast we observe a duality by showing that, when learning from informant, requiring the learner to eventually output exactly one correct enumeration procedure is as powerful as allowing any finite number of correct descriptions in the limit. Furthermore, even facing all appropriate learning restrictions at hand gives us more learning power for $b=\infty$, known as behaviorally correct ($\Bc$) learning.
In particular, we obtain for all $b\in\N_{>0}$
$$[\Inf\Lim] = \ldots = [\Inf\Lim_b] = [\Inf\Lim_{b+1}] = \ldots \subsetneq [\Inf\Bc].$$
We also compare learning settings in which both $a$ and $b$ do not take their standard values $0$ and $1$, respectively.

\subsection{More Connections to Prior Research}

In contrast to our observations, it has been shown in \cite{angluin1980inductive} that requiring a conservative learning process is a restriction when learning from text. Further, this is equivalent to cautious learning as shown in \cite{kotzing2016map}.
That monotonic learning is restrictive and incomparable to both of them in the text learning setting follows from \cite{lange1996monotonic}, \cite{kinber1995language}, \cite{Jai-Sha:j:98} and \cite{kotzing2016map}.
Further, when learning from text, strong monotonicity is again the most restrictive assumption by \cite{lange1996monotonic}.
Strong decisiveness is restrictive, see \cite{Bal-Cas-Mer-Ste-Wie:j:08}, and further is restricted by cautiousness/conservativeness on the one hand and monotonicity on the other hand by \cite{kotzing2016map}.
In the latter visualizations and a detailed discussion are provided.

When the learner does not have access to the order of presentation but knows the number of samples, the map remains the same as observed in \cite{kotzing2016towards}.

In case the learner makes its decisions only based on the set of presented samples and ignores any information about the way it is presented, it is called \emph{set-driven} ($\Sd$). For such set-driven learners, when learning from text, conservative, strongly decisive and cautious learning are no longer restrictive and the situation with monotonic and strong monotonic learning remains unchanged by \cite{kinber1995language} and \cite{kotzing2016map}.

We observe that for delayable informant learning all three kinds of learners yield the same map. Thus, our results imply that negative information compensates for the lack of information set-driven learners have to deal with.

\medskip
\cite{Gol:j:67} was already interested in the above mentioned normal forms and proved that they can be assumed without loss of generality in the basic setting of pure $\Lim$-learning, whereas our results apply to all delayable learning success criteria. 

\medskip
The name ``delayability'' refers to tricks in order to delay mind changes of the learner which were used to obtain polynomial computation times for the learners hypothesis updates as discussed in \cite{Pit:c:89} and \cite{Cas-Koe:c:09}.
Moreover, it should not be confused with the notion of $\delta$-delay, \cite{akama2008consistent}, which allows satisfaction of the considered learning restriction $\delta$ steps later than in the un-$\delta$-delayed version.

\medskip
In \cite{STL1} several restrictions for learning from informant are analyzed and mentioned that cautious learning is a restriction to learning power; we extend this statement with our Proposition~\ref{InfMonNotCautEx} in which we give one half of the answer to the second question above by providing a family of languages not cautiously but monotonically $\Lim$-learnable from informant.

\medskip
Furthermore, \cite{STL1} consider a version of \emph{conservativeness} where mind changes are only allowed if there is \emph{positive} data contradicting the current hypothesis, which they claim to restrict learning power.
In this paper, we stick to the more common definition in \cite{Blu-Blu:j:75} and \cite{Bar:c:77:aut-func-prog}, according to which mind changes are allowed also when there is negative data contradicting the current hypothesis.

\subsection{Outline}

In Section~\ref{sec:setting} the setting of learning from informant is formally introduced by transferring fundamental definitions and ---as far as possible--- observations from the setting of learning from text.
In Section~\ref{sec:WlogTotalMethodical} in order to derive the entire map of pairwise relations between delayable $\Lim$-learning success criteria, normal forms and a regularity property for such learning from informant are provided. Further, consistent learning is being investigated.
In Section~\ref{sec:map_delayable} we answer the questions above and present all pairwise relations of learning criteria in Theorem~\ref{InfExMap}.
In Section \ref{sec:infvstxt} we generalize the result in \cite{Gol:j:67}, we already gave a proof-sketch for, namely $\Lim$-learning from text to be harder than $\Lim$-learning from informant.
In Section~\ref{sec:hierarchies} we provide the aforementioned anomalous hierarchy and vacillatory duality.

\medskip
We kept every section as self-contained as possible. Unavoidably, all sections build on Section~\ref{sec:setting}.
Additionally, Section~\ref{sec:map_delayable} builds on Section~\ref{sec:WlogTotalMethodical}.

\section{Informant Learning}
\label{sec:setting}
We formally introduce the notion of an informant and transfer concepts and fundamental results from the setting of learning from text to learning from informant.
This includes the learner itself, convergence criteria, locking sequences, learning restrictions and success criteria as well as a compact notation for comparing different learning settings.
In the last subsection delayability as the central property of learning restrictions and learning success criteria is formally introduced.

\medskip
As far as possible, notation and terminology on the learning theoretic side follow \cite{Jai-Osh-Roy-Sha:b:99:stl2}, whereas on the computability theoretic side we refer to \cite{Odi:b:99}.

\medskip
We let $\N$ denote the \emph{natural numbers} including $0$ and write $\infty$ for an \emph{infinite cardinality}.
Moreover, for a function $f$ we write $\dom(f)$ for its \emph{domain} and $\ran(f)$ for its \emph{range}.
If we deal with (a subset of) a cartesian product, we are going to refer to the \emph{projection functions} to the first or second coordinate by $\pr_1$ and $\pr_2$, respectively.
For sets $X, Y$ and $a \in \N$ we write $X =^a Y$, if \emph{$X$ equals $Y$ with $a$ anomalies}, i.e., $| (X\setminus Y) \cup (Y\setminus X) | \leq a$, where $|.|$ denotes the \emph{cardinality function}.
In this spirit we write $X =^\ast Y$, if there exists some $a \in \N$ such that $X = ^a Y$.
Further, $\Seq{X}$ denotes the \emph{finite sequence}s over $X$ and $\ISeq{X}$ stands for the \emph{countably infinite sequence}s over $X$.
Additionally, $\IfiSeq{X} := \Seq{X} \cup \ISeq{X}$ denotes the set of all \emph{countably finite or infinite sequence}s over $X$. 
For every $f \in \IfiSeq{X}$ and $t \in \N$, we let
$f[t] := \{ (s,f(s)) \mid s < t \}$ denote the \emph{restriction of $f$ to $t$}.
Finally, for sequences $\sigma, \tau \in \Seq{X}$ their concatenation is denoted by $\sigma\concat\tau$ and we write $\sigma \inseg \tau$, if $\sigma$ is an initial segment of $\tau$, i.e., there is some $t \in \N$ such that $\sigma = \tau[t]$.
In our setting, we typically have 
$X=\N \times \{0,1\}$.
We denote by $\partialFn$ and $\totalFn$ the set of all partial functions $f: \dom(f) \subseteq \Seq{\Nttwo} \to \N$ and total functions $f: \Seq{\Nttwo} \to \N$, respectively.

\medskip
Let $L \subseteq \N$.
If $L$ is recursively enumerable, we call $L$ a \emph{language}.
In case its characteristic function is computable, we say it is a \emph{recursive language}.
Moreover, we call $\CalL \subseteq \Pow(\N)$ a \emph{collection of (recursive) languages}, if every $L \in \CalL$ is a (recursive) language.
In case there exists an enumeration $\{ L_\xi \mid \xi \in \Xi \}$ of $\CalL$, where $\Xi \subseteq \N$ is recursive
and a computable function $f$ with $\ran(f) \subseteq \{0,1\}$ such that 
$x \in L_\xi \Leftrightarrow f(x,\xi)=1$ for all $\xi \in \Xi$ and $x \in \N$,
we say $\CalL$ is an \emph{indexable family of recursive languages}. By definition indexable families are collections of recursive languages with a uniform decision procedure.

Further, we fix a programming system $\varphi$ as introduced in \cite{Roy-Cas:b:94}. Briefly, in the $\varphi$-system, for a natural number $p$, we denote by $\varphi_p$ the partial computable function with program code $p$.
We call $p$ an \emph{index} for $W_p := \dom(\varphi_p)$.
For a finite set $X\subseteq\N$ we denote by $\ind(X)$ a canonical index for $X$.
In reference to a Blum complexity measure, for all $p, t \in \N$, we denote by $W^t_p \subseteq W_p$ the recursive set of all natural numbers less or equal to $t$, on which the machine executing $p$ halts in at most $t$ steps.
Moreover, by s-m-n we refer to a well-known recursion theoretic observation, which gives finite and infinite recursion theorems, like Case's Operator Recursion Theorem $\ORT$, \cite{case1974periodicity}.
Inuitively, it states that for every recursive operator there is a computable function that is a fixed point of the action of the operator on the $\varphi$-system.
Formally, a 1-1 version of this result reads as follows.

\begin{11ORT}[\cite{Koe:th:09}]
\label{11ORT}
Let $\Theta: \partialCp \to \partialCp$ be a computable operator, namely a function mapping partial computable functions to partial computable functions.
Then there is a 1-1 computable function $h \in \partialCp$ such that
$\forall n, x \left( \varphi_{h(n)}(x) = \Theta(h)(n, x) \right).$
\end{11ORT}

\medskip
For our purposes the operator $\Theta$ will always be implicit.
The first application of ORT is in Proposition~\ref{prop:TotalInfConsEx} and it occurs in many different variants in other proofs.
For further intuitions see for example \cite{Cas:j:94:self}.

Finally, we let $H = \{\,p \in \N\mid \varphi_p(p)\!\!\downarrow\,\}$ denote the halting problem.

\subsection{informant and Learners}


Intuitively, for any natural number $x$ an \emph{informant for a language $L$} answers the question whether $x \in L$ in finite time.
More precisely, for every natural number $x$ the informant $I$ has either $(x,1)$ or $(x,0)$ in its range, where the first is interpreted as $x\in L$ and the second as $x \notin L$, respectively.

\begin{definition}
\label{def:informant}
\begin{enumerate}
\item Let $f \in \IfiSeq{(\Nttwo)}$. We denote by
\begin{align*}
\ps(f) &:= \{ y \in \N \mid \exists x \in \N \colon \pr_1(f(x))=y \wedge \pr_2(f(x))=1 \}, \\
\ng(f) &:= \{ y \in \N \mid \exists x \in \N \colon \pr_1(f(x))=y \wedge \pr_2(f(x))=0 \}
\end{align*}
the sets of all natural numbers, about which $f$ gives some positive or negative information, respectively.
\item Let $L$ be a language. We call every function $I: \N \to \Nttwo$ such that $\ps(I) \cup \ng(I) = \N$ and $\ps(I) \cap \ng(I) = \varnothing$ an \emph{informant}. Further, we denote by $\CalI$ the set of all informant and
the \emph{set of all informant for the language $L$} is defined as
$$\CalI(L) := \{ I \in \CalI \mid \ps(I) = L \}.$$
\item Let $I$ be an informant. If for every time $t \in \N$ reveals information about $t$ itself, for short $\pr_1(I(t))=t$, we call $I$ a \emph{canonical informant}.
\end{enumerate}
\end{definition}

It is immediate, that $\ng(I)=\N\setminus L$ for every $I \in \CalI(L)$.
In \cite{Gol:j:67} a canonical informant is referred to as \emph{methodical informant}.

\bigskip
We employ Turings model for human computers which is the foundation of all modern computers to model the processes in human and machine learning.

\begin{definition}
\label{def:learner}
A \emph{learner} is a (partial) computable function
$$M: dom(M) \subseteq \Seq{(\Nttwo)} \to \N.$$
\end{definition}

The set of all partial computable functions $M: \dom(M) \subseteq \Seq{\Nttwo} \to \N$ and total computable functions $M: \Seq{\Nttwo} \to \N$ are denoted by $\partialCp$ and $\totalCp$, respectively.

\subsection{Convergence Criteria and Locking Sequences}

Convergence criteria tell us what quality of the approximation and syntactic accuracy of the learners' eventual hypotheses are necessary to call learning successful.
Further, we proof that learning success implies the existence of sequences on which the learner is locked in a way corresponding to the convergence criterion. We will use locking sequences to show that a collection of languages cannot be learned in a certain way.

\begin{definition}
\label{def:InfConvergence}
Let $M$ be a learner and $\CalL$ a collection of languages. Further, let $a \in \N\cup\{\ast\}$ and $b \in \N_{>0}\cup\{\ast,\infty\}$.
\begin{enumerate}
\item Let $L \in \CalL$ be a language and $I \in \Inf(L)$ an informant for $L$ presented to $M$.
\begin{enumerate}
\item We call $h=(h_t)_{t\in\N} \in \ISeq{\N}$, where $h_t := M(I[t])$ for all $t \in \N$, the \emph{learning sequence of $M$ on $I$}.
\item $M$ \emph{learns $L$ from $I$ with $a$ anomalies and vacillation number $b$ in the limit}, for short $M$ $\Lim^a_b$-learns $L$ from $I$ or $\Lim_b^a(M,I)$, if there is a time $t_0 \in \N$ such that $|\,\{\, h_t \mid t \geq t_0 \,\}\,| \leq b$ and for all $t \geq t_0$ we have $W_{h_t}=^a L$.
\end{enumerate}
\item $M$ \emph{learns $\CalL$ with $a$ anomalies and vacillation number $b$ in the limit}, for short $M$ $\Lim^a_b$-learns $\CalL$, if $\Lim^a_b(M,I)$ for every $L\in\CalL$ and every $I \in \Inf(L)$.
\end{enumerate}
\end{definition}

The intuition behind \textit{(i)(b)} is that, sensing $I$, $M$ eventually only vacillates between at most $b$-many hypotheses, where the case $b = \ast$ stands for eventually finitely many different hypotheses.
In convenience with the literature, we omit the superscript $0$ and the subscript $1$.


\medskip
$\Lim$-learning, also known as \emph{explanatory learning}, is the most common definition for successful learning and corresponds to the notion of identifiability in the limit by \cite{Gol:j:67}, where the learner eventually decides on one correct hypothesis.
On the other end of the hierarchy of convergence criteria is \emph{behaviorally correct learning}, for short \emph{$\Bc$-} or $\Lim_\infty$-learning, which only requires the learner to be eventually correct, but allows infinitely many syntactically different hypotheses in the limit.
Behaviorally correct learning was introduced in  \cite{Osh-Wei:j:82:criteria}.
The general definition of $\Lim^a_b$-learning for $a \in \N\cup\{\ast\}$ and $b \in \N_{>0}\cup\{\ast\}$ was first mentioned in \cite{case1999power}.

In our setting, we also allow $b=\infty$ and subsume all $\Lim^a_b$ under the notion of a \emph{convergence criterion},
since they determine in which semi-topological sense the learning sequence needs to have $L$ as its limit, in order to succeed in learning $L$.

\bigskip
In the following we transfer an often employed observation in \cite{Blu-Blu:j:75} to the setting of learning from informant and generalize it to all convergence criteria introduced in Definition~\ref{def:InfConvergence}.

\begin{definition}
\label{def:lockingsequence}
Let $M$ be a learner, $L$ a language and $a\in\N\cup\{\ast\}$ as well as $b\in\N_{>0}\cup\{\ast,\infty\}$. We call $\sigma \in \Seq{(\Nttwo)}$ a \emph{$\Lim^a_b$-locking sequence for $M$ on $L$}, if $\Comp(\sigma,L)$ and
\begin{align*}
\exists D \subseteq \N \:\left(\,\right. &|D|\leq b \:\wedge\: \forall \tau \in \Seq{(\Nttwo)} \\
&\left. \left(\, \Comp(\tau,L) \Rightarrow \left( M(\sigma\concat\tau)\!\!\downarrow \wedge \,W_{M(\sigma\concat\tau)}=^aL \wedge M(\sigma\concat\tau)\in D \,\right)\right) \,\right)
\end{align*}
Further, a \emph{locking sequence for $M$ on $L$} is a $\Lim$-locking sequence for $M$ on $L$.
\end{definition}

Intuitively, the learner $M$ is locked by the sequence $\sigma$ onto the language $L$ in the sense that no presentation consistent with $L$ can circumvent $M$ guessing admissible approximations to $L$ and additionally all guesses based on an extension of $\sigma$ are captured by a finite set of size at most $b$.

\smallskip
Note that the definition implies $M(\sigma)\!\!\downarrow$, $W_{M(\sigma)}=^aL$ and $M(\sigma)\in D$.

\begin{lemma}
\label{lem:lockingsequence}
Let $M$ be a learner, $a\in\N\cup\{\ast\}$, $b\in\N_{>0}\cup\{\ast,\infty\}$ and $L$ a language $\Lim^a_b$-identified by $M$. Then there is a $\Lim^a_b$-locking sequence for $M$ on $L$.
\end{lemma}
\begin{proof}
This is a contradictory argument.
Without loss of generality $M$ is defined on $\varnothing$.
Assume towards a contradiction for every $\sigma$ with $\Comp(\sigma,L)$, $M(\sigma)\!\!\downarrow$ and $W_{M(\sigma)}=^aL$ and for every finite $D \subseteq \N$ with at most $b$ elements there exists a sequence $\tau^D_\sigma \in \Seq{(\Nttwo)}$ with
\begin{align*}
\Comp(\tau^D_\sigma,L) \wedge \left(\, M(\sigma\concat\tau^D_\sigma)\!\!\uparrow \vee \,\neg W_{M(\sigma\concat\tau^D_\sigma)}=^aL \vee M(\sigma\concat\tau^D_\sigma)\notin D \,\right).
\end{align*}
Let $I_L$ denote the canonical informant for $L$.
We obtain an informant for $L$ on which $M$ does not $\Lim^a_b$-converge by letting
\begin{align*}
I &:= \bigcup_{n\in\N}\sigma_n, \text{ with} \\
\sigma_0 &:= I_L[1],\\
\sigma_{n+1} &:= \sigma_n\concat\tau^{D_n}_{\sigma_n}{\concat}I_L(n+1)
\end{align*}
for all $n \in \N$, where in $D_n :=\{\,M(\sigma_i^-)\mid \max\{0,n-b+1\} \leq i \leq n\,\}$ we collect $M$'s at most $b$-many last relevant hypotheses.
Since $I$ is an informant for $L$ by having interlaced the canonical informant for $L$, the learner $M$  $\Lim^a_b$-converges on $I$.
Therefore, let $n_0$ be such that for all $t$ with $\sigma_{n_0}^- \inseg I[t]$ we have $h_t\!\!\downarrow$ and $W_{h_t}=^a L$. Then certainly $\{\,M(\sigma_i^-)\mid n_0\leq i\leq n_0+b\,\}$ has cardinality $b+1$, a contradiction.
\end{proof}

Obviously, an appropriate version also holds when learning from text is considered.


\subsection{Learning Success Criteria}
\label{subsec:learningsuccess}

We list the most common requirements that combined with a convergence criterion define when a learning process is considered successful. For this we first recall the notion of consistency of a sequence with a set according to \cite{Blu-Blu:j:75} and  \cite{Bar:c:77:aut-func-prog}.

\begin{definition}
\label{def:consistent}
Let $f \in \IfiSeq{(\Nttwo)}$ and $A \subseteq \N$. We define
\begin{align*}
\Comp(f,A) \quad &:\Leftrightarrow \quad \ps(f) \subseteq A \;\wedge\; \ng(f) \subseteq \N\setminus A
\end{align*}
and say \emph{$f$ is consistent with $A$}.
\end{definition}


The choice of learning restrictions in the following definition is justified by prior investigations of the corresponding criteria, when learning from text, see \cite{kotzing2016map}, \cite{kotzing2016towards} and \cite{jain2016role}. 

\begin{definition}
\label{def:learningrestrictions}
Let $M$ be a learner, $I \in \CalI$ an informant and $h=(h_t)_{t\in\N} \in \ISeq{\N}$ the learning sequence of $M$ on $I$.
We write
\begin{enumerate}
        \item $\Cons(M,I)$ (\cite{angluin1980inductive}), if $M$ is \emph{consistent on $I$}, i.e., for all $t$
        $$\Comp(I[t],W_{h_t}).$$
        \item $\Conv(M,I)$ (\cite{angluin1980inductive}), if $M$ is \emph{conservative on $I$}, i.e., for all $s, t$ with $s \leq t$
        $$\Comp(I[t],W_{h_s}) \;\Rightarrow\; h_s = h_t.$$
        \item  $\Dec(M,I)$ (\cite{Osh-Sto-Wei:j:82:strategies}),
        if $M$ is \emph{decisive on $I$}, i.e.,
        for all $r, s, t$ with $r \leq s \leq t$
        $$W_{h_r} = W_{h_t} \;\Rightarrow\; W_{h_r} = W_{h_s}.$$
        \item $\Caut(M,I)$ (\cite{STL1}),
		if $M$ is \emph{cautious on $I$}, i.e.,
		for all $s, t$ with $s \leq t$
        $$\neg W_{h_t} \subsetneq W_{h_s}.$$
		\item $\WMon(M,I)$ (\cite{j-mniifp-91},\cite{Wie:c:91}), if $M$ is \emph{weakly monotonic on $I$}, i.e., for all $s, t$ with $s \leq t$\\
        $$\Comp(I[t], W_{h_s})
        \;\Rightarrow\; W_{h_s} \subseteq W_{h_t}.$$
        \item $\Mon(M,I)$ (\cite{j-mniifp-91},\cite{Wie:c:91}), if $M$ is \emph{monotonic on $I$}, i.e., for all $s, t$ with $s \leq t$
        $$W_{h_s} \cap \ps(I) \subseteq W_{h_t}\cap\ps(I).$$
        \item  $\SMon(M,I)$ (\cite{j-mniifp-91},\cite{Wie:c:91}), if $M$ is \emph{strongly monotonic on $I$}, i.e., for all $s, t$ with $s \leq t$
        $$W_{h_s} \subseteq W_{h_t}.$$
        \item  $\NU(M,I)$ (\cite{Bal-Cas-Mer-Ste-Wie:j:08}), if $M$ is \emph{non-U-shaped on $I$}, i.e., for all $r, s, t$ with $r \leq s \leq t$
        $$W_{h_r} = W_{h_t} = \ps(I) \;\Rightarrow\; W_{h_r} = W_{h_s}.$$
        \item  $\SNU(M,I)$ (\cite{Cas-Moe:j:11:optLan}), if $M$ is \emph{strongly non-U-shaped on $I$}, i.e., for all $r, s, t$ with $r \leq s \leq t$
        $$W_{h_r} = W_{h_t} = \ps(I) \;\Rightarrow\; h_r = h_s.$$
        \item  $\SDec(M,I)$ (\cite{kotzing2016map}), if $M$ is \emph{strongly decisive on $I$}, i.e., for all $r, s, t$ with $r \leq s \leq t$
        $$W_{h_r} = W_{h_t} \;\Rightarrow\; h_r = h_s.$$

\end{enumerate}%
\end{definition}


The following lemma states the implications between almost all of the above defined learning restrictions, which form the foundation of our research. Figure~\ref{MapInfDelayableEx} 
includes the resulting backbone, which is slightly different from the one for learning from text, since $\WMon$ does not necessarily imply $\NU$ in the context of learning from informant. 

\begin{lemma}
\label{lem:delayablebackbone}
Let $M$ be a learner and $I \in \CalI$ an informant. Then
\begin{enumerate}
\item $\Conv(M,I)$ implies $\SNU(M,I)$ and $\WMon(M,I)$.
\item $\SDec(M,I)$ implies $\Dec(M,I)$ and $\SNU(M,I)$.
\item $\SMon(M,I)$ implies $\Caut(M,I), \Dec(M,I), \Mon(M,I)$ as well as $\WMon(M,I)$.
\item $\Dec(M,I)$ and $\SNU(M,I)$ imply $\NU(M,I)$.
\item $\WMon(M,I)$ does \emph{not} imply $\NU(M,I)$.
\end{enumerate}
\end{lemma}
\begin{proof}
Verifying the claimed implications is straightforward.
In order to verify \textit{(v)}, consider $L = 2\N$. 
Fix $p, q \in \N$ such that $W_p = 2\N \cup \{1\}$ and $W_q = 2\N$ and define the learner $M$ for all $\sigma \in \Seq{\Nttwo}$ by
\begin{align*}
M(\sigma) &= \begin{cases}
p, &\text{if } 1 \in \ng(\sigma) \wedge 2 \notin \ps(\sigma); \\
q, &\text{otherwise}.
\end{cases}
\intertext{In order to prove $\WMon(M,I)$ for every $I \in \Inf(L)$,
let $I$ be an informant for $L$ and
$\simu_I(x):=\min\{t \in \N \mid \pr_1(I(t))=x \}$,
i.e., $\simu_I(1)$ and $\simu_I(2)$ denote the first occurance of $(1,0)$ and $(2,1)$ in $\ran(I)$, respectively.
Then we have for all $t \in \N$}
W_{h_t}&=
\begin{cases}
2\N\cup\{1\}, &\text{if } \simu_I(1) < t \leq \simu_I(2); \\
2\N, &\text{otherwise}.
\end{cases}
\end{align*}
We have $W_{h_s}=W_{M(I[s])}=2\N\cup\{1\}$ as well as $1 \in \ng(I[t])$ for all $s,t\in \N$ with $\simu_I(1) < s \leq \simu_I(2)$ and $t > \simu_I(2)$.
Therefore, $\neg \Comp(I[t],W_{h_s})$ because of $\ng(I[t])\not \subseteq \N\setminus W_{h_s}$.
We obtain $\WMon(M,I)$ since whenever $s \leq t$ in $\N$ are such that $\Comp(I[t],W_{h_s})$, we know that $W_{h_s} = 2\N\cup\{1\}$ can only hold if likewise $\simu_I(1) < t \leq \simu_I(2)$ and hence $W_{h_t}=2\N\cup\{1\}$, which yields $W_{h_s}\subseteq W_{h_t}$. Furthermore, if $W_{h_s}=2\N$ all options for $W_{h_t}$ satisfy $W_{h_s}\subseteq W_{h_t}$.
Otherwise, in case $M$ observes the canonical informant $I$ for $L$, we have $W_{h_0}=W_{h_1}=2\N$, $W_{h_2}=2\N\cup\{1\}$ and $W_{h_t}=2\N$ for all $t > 2$, which shows $\neg \NU(M,I)$.
\end{proof}

By the next definition, in order to characterize what successful learning means, we choose a convergence criterion from Definition~\ref{def:InfConvergence} and may pose additional learning restrictions from Definition~\ref{def:learningrestrictions}.

\begin{definition}
\label{def:learningsuccesscriterion}
Let $\mathbf{T}:=\partialFn\times\Inf$ denote the whole set of pairs of possible learners and informant.
We denote by
\begin{align*}
    \Delta := \{ \,&\Caut, \Cons, \Conv, \Dec, \SDec, \\
            &\WMon, \Mon,\SMon,\NU,\SNU, \mathbf{T} \,\}
\end{align*}
the set of \emph{admissible learning restrictions} and by
$$\Gamma := \{\, \Lim^a_b \mid a \in \N\cup\{\ast\} \:\wedge\: b \in \N_{>0}\cup\{\ast,\infty\} \,\}$$ the set of \emph{convergence criteria}.
Further, if $$\beta \in \{ \: \bigcap_{i=0}^n \delta_i \cap \gamma \mid
n \in \N, \forall i \leq n \,(\delta_i \in \Delta) \text{ and } \gamma \in \Gamma \:\} \subseteq \partialFn \times \CalI,$$
we say that $\beta$ is a \emph{learning success criterion}.
\end{definition}

Note that every convergence criterion is indeed a learning success criterion by letting $n=0$ and $\delta_0=\mathbf{T}$, where the latter stands for no restriction. In the literature convergence criteria are also called identificaton criteria and then denoted by $I$ or $ID$.

\smallskip
We refer to all $\delta \in \{\Caut,\Cons,\Dec,\Mon,\SMon,\WMon,\NU,\mathbf{T}\}$ also as \emph{semantic} learning restrictions, as they allow for proper semantic convergence.

\subsection{Comparing the Learning Power of Learning Settings}

In order to state observations about how two ways of defining learning success relate to each other, the learning power of the different settings is encapsulated in notions $[\alpha\Inf\beta]$ defined as follows.

\begin{definition}
\label{def:learningcriterium}
Let $\alpha\subseteq\partialCp$ be a property of partial computable functions from the set $\Seq{(\Nttwo)}$ to $\N$ and $\beta$ a learning success criterion.
We denote by $[\alpha\Inf\beta]$ the set of all collections of languages that are $\beta$-learnable from informant by a learner $M$ with the property $\alpha$.

In case the learner only needs to succeed on canonical informant, we denote the corresponding set of collections of languages by $[\alpha\Inf_\can\beta]$. 
\end{definition}

In the learning success criterion at position $\beta$, the learning restrictions to meet are denoted in alphabetic order, followed by a convergence criterion.

At position $\alpha$, we restrict the set of admissible learners by requiring for example totality. The properties stated at position $\alpha$ are \emph{independent of learning success}.

\medskip
For example, a collection of languages $\CalL$ lies in $[\totalCp\Inf_\can\Conv\SDec\Lim]$ if and only if there is a total learner $M$ conservatively, strongly decisively $\Lim$-learning every $L \in \CalL$ from canonical informant. The latter means that for every canonical informant $I$ for some $L \in\CalL$ we have $\Conv(M,I)$, $\SDec(M,I)$ and $\Lim(M,I)$.

\medskip
Note that it is also conventional to require $M$'s hypothesis sequence to fulfill certain learning restrictions, not asking for the success of the learning process.
For instance, we are going to show that there is a collection of languages $\CalL$ such that:
\begin{itemize}
\item there is a learner which behaves consistently on all $L\in\CalL$ and $\Lim$-learns all of them, for short $\CalL \in [\Inf\Cons\Lim]$.
\item there is no learner which $\Lim$-learns every $L \in \CalL$ and behaves consistently on all languages, for short $\CalL \notin [\Cons\Inf\Lim]$.
\end{itemize}
The existence of $\CalL$ is implicit when writing $[\Cons\Inf\Lim] \subsetneq [\Inf\Cons\Lim]$.

\bigskip
This notation makes it also possible to distinguish the mode of information presentation. If the learner observes the language as solely positive information, we write $[\alpha\Txt\beta]$ for the collections of languages $\beta$-learnable by a learner with property $\alpha$ from text. Of course for $\alpha$ and $\beta$ the original definitions for the setting of learning
from text have to be used.
All formal definitions for learning from text can be found in \cite{kotzing2014map}.

\subsection{Delayability}

We now introduce a property of learning restrictions and learning success criteria, which allows general observations, not bound to the setting of $\Lim$-learning, since it applies to all of the learning restrictions introduced in Definition~\ref{def:learningrestrictions} except consistency.

\begin{definition}
\label{def:delayable}
Denote the set of all unbounded and non-decreasing functions by $\Simu$, i.e., $\Simu := \{ \,\simu: \N \to \N \mid \forall x \in \N \,\exists t \in \N \colon \simu(t) \geq x \text{ and } \forall t \in \N \colon \simu(t+1) \geq \simu(t) \,\}.$
Then every $\simu \in \Simu$ is a so called \emph{admissible simulating function}.

\smallskip
A predicate $\beta \subseteq \partialFn \times \CalI$
is \emph{delayable}, if for all $\simu \in \Simu$, all $I, I' \in \CalI$ and all
partial functions $M, M'\in \partFn$ holds:
Whenever we have $\ps(I'[t])  \supseteq \ps(I[\simu(t)])$, $\ng(I'[t]) \supseteq \ng(I[\simu(t)])$ and $M'(I'[t]) = M(I[\simu(t)])$ for all $t \in \N$, from $\beta(M,I)$ we can conclude $\beta(M', I')$.
\end{definition}

The unboundedness of the simulating function guarantees $\ps(I)=\ps(I')$ and \linebreak[3] $\ng(I)=\ng(I')$.

\medskip
In order to give an intuition for delayability, think of $\beta$ as a learning restriction or learning success criterion and imagine $M$ to be a learner.
Then $\beta$ is delayable if and only if it carries over from $M$ together with an informant $I$ to all learners $M'$ and informant $I'$ representing a delayed version of $M$ on $I$. 
More concretely, as long as the learner $M'$ conjectures $h_{\simu(t)}=M(I[\simu(t)])$ at time $t$ and has, in form of $I'[t]$, at least as much data available as was used by $M$ for this hypothesis, $M'$ with $I'$ is considered a delayed version of $M$ with $I$.

\bigskip
The next result guarantees that arguing with the just defined properties covers all of the considered learning restrictions but consistency.

\begin{lemma}
\label{lem:delayable}
\begin{enumerate}
\item Let $\delta \in \Delta$ be a learning restriction. Then $\delta$ is delayable if and only if $\delta \neq \Cons$.
\item Every convergence criterion $\gamma \in \Gamma$ is delayable.
\item The intersection of finitely many delayable predicates on 
$\partialFn \times \CalI$ is again delayable.
Especially, every learning success criterion $\beta = \bigcap_{i=0}^n \delta_i \cap \gamma$ with $\delta_i \in \Delta\setminus\{\Cons\}$ for all $i \leq n$ and $\gamma \in \Gamma$, $\beta$ is delayable.
\end{enumerate}
\end{lemma}
\begin{proof}
We approach $(i)$ by showing, that $\Cons$ is not delayable.
To do so, consider $\simu \in \Simu$ with $\simu(t) := \lfloor \frac{t}{2}\rfloor$, $I, I' \in \CalI$ defined by $I(x) := (\lfloor \frac{x}{2}\rfloor, \mathbb{1}_{2\N}(\lfloor \frac{x}{2}\rfloor))$ and $I'(x):=(x, \mathbb{1}_{2\N}(x))$, where $\mathbb{1}_{2\N}$ stands for the characteristic function of all even natural numbers.
By s-m-n there are learners $M$ and $M'$ such that for all $\sigma \in \Seq{(\Nttwo)}$ 
\begin{align*}
W_{M(\sigma)} &= \{ x \in \N \mid (x \text{ even }  \wedge x \leq \lfloor \frac{|\sigma|}{2} \rfloor) \vee (x \text{ odd }  \wedge x > \lfloor \frac{|\sigma|}{2} \rfloor)\} \\
W_{M'(\sigma)} &= \{ x \in \N \mid (x \text{ even } \wedge x \leq \lfloor \frac{|\sigma|}{4} \rfloor) \vee (x \text{ odd } \wedge x > \lfloor \frac{|\sigma|}{4} \rfloor)\}.
\end{align*}
Further, $\Cons(M,I)$ is easily verified since for all $t \in \N$
\begin{align*}
\ps(I[t]) &= \{ x \in \N \mid x \text{ even } \wedge x \leq \lfloor \frac{t-1}{2} \rfloor\} \subseteq W_{M(I[t])} \\
\ng(I[t]) &= \{ x \in \N \mid x \text{ odd } \wedge x \leq \lfloor \frac{t-1}{2} \rfloor)\} \subseteq \N\setminus W_{M(I[t])} \\
\intertext{but on the other hand $\neg \Cons(M',I')$ since for all $t > 2$}
\ps(I'[t]) &= \{ x \in \N \mid x \text{ even } \wedge x < t\} \\
&\not \subseteq \{ x \in \N \mid (x \text{ even }  \wedge x \leq \lfloor \frac{t}{4} \rfloor) \vee (x \text{ odd }  \wedge x > \lfloor \frac{t}{4} \rfloor)\} = W_{M'(I'[t])}.
\end{align*}

The remaining proofs for $(i)$ and $(ii)$ are straightforward. Basically, for $\Dec,$ $\SDec,$ $\SMon$ and $\Caut$, the simulating function $\simu$ being non-decreasing and $M'(I'[t])=M(I[\simu(t)])$ for all $t \in \N$ would suffice, while for $\NU, \SNU$ and $\Mon$ one further needs that the informant $I$ and $I'$ satisfy $\ps(I)=\ps(I')$.
The proof for $\WMon$ and $\Conv$ to be delayable, requires all assumptions, but $\simu$'s unboundedness.
Last but not least, in order to prove that every convergence criterion $\gamma = \Lim^a_b$, for some $a \in \N\cup\{\ast\}$ and $b \in \N_{>0}\cup\{\ast,\infty\}$, carries over to delayed variants, one essentially needs both characterizing properties of $\simu$ and of course $M'(I'[t])=M(I[\simu(t)])$. Finally, $(iii)$ is obvious.
\end{proof}

\section{Delayability vs. Consistency: \\ Canonical informant and Totality}
\label{sec:WlogTotalMethodical}
In order to facilitate smooth proofs later on, we discuss normal forms for learning from informant. First, we consider the notion of set-drivenness. 
In Lemma~\ref{lem:canInfSdDelayableEx} we show for delayable learning success criteria, that every collection of languages that is learnable from canonical informant is also learnable by a set-driven learner from arbitrary informant. By Proposition~\ref{prop:canInfConsEx} this does not hold for consistent $\Lim$-learning. This also implies that consistency is a restriction when learning from informant. Moreover, in Lemma~\ref{lem:TotalInfDelayableEx} we observe that only considering total learners does not alter the learnability of a collection of languages in case of a delayable learning success criterion.
This does not hold for consistent $\Lim$-learning by Proposition~\ref{prop:TotalInfConsEx}.

\subsection{Set-driven Learners and Canonical informant}

We start by formally capturing the intuition for a learner being set-driven, given in the introduction.

\begin{definition}[\cite{Wex-Cul:b:80}]
\label{def:Restrictions}
A learner $M$ is \emph{set-driven}, for short $\Sd(M)$, if for all $\sigma, \tau \in \Seq{\Nttwo}$
$$(\,\ps(\sigma)=\ps(\tau) \wedge \ng(\sigma)=\ng(\tau)\,) \;\Rightarrow\; M(\sigma)=M(\tau).$$
\end{definition}

\cite{Sch:th:84} and \cite{Ful:th:85} showed that set-drivenness is a restriction when learning only from positive information and also the relation between the learning restrictions differ as observed in \cite{kotzing2016map}.

\medskip
In the next Lemma we observe that, by contrast, set-drivenness is not a restriction in the setting of learning from informant. Concurrently, we generalize \cite{Gol:j:67}'s observation, stating that considering solely canonical informant to determine learning success does not give more learning power, to arbitrary delayable learning success criteria.

\begin{lemma}
\label{lem:canInfSdDelayableEx}
Let $\beta$ be a delayable learning success criterion. Then 
\begin{equation*}
        [\Inf_\can\beta]=[\Sd\Inf\beta].
\end{equation*}
\end{lemma}
\begin{proof}
Clearly, we have $[\Inf_\can\beta]\supseteq[\Sd\Inf\beta]$.
For the other inclusion, let $\CalL$ be $\beta$-learnable by a learner $M$ from canonical informant.
We proceed by formally showing that rearranging the input on the initial segment of $\N$, we already have complete information about at that time, is an admissible simulation in the sense of Definition~\ref{def:delayable}.
Let $L \in \CalL$ and $I' \in \CalI(L)$.
For every $f \in \IfiSeq{(\Nttwo)}$, thus especially for $I'$ and all its initial segments, we define $\simu_{f} \in \Simu$ for all $t$ for which $f[t]$ is defined, by
$$\simu_{f}(t)=\sup\{ x \in \N \mid \forall w < x \colon w \in \ps(f[t])\cup\ng(f[t]) \},$$
i.e., the largest natural number $x$ such that for all $w < x$ we know, whether $w \in \ps(f)$.
In the following $f$ will either be $I'$ or one of its initial segments, which in any case ensures $\ps(f[t]) \subseteq L$ for all appropriate $t$.
By construction, $\simu_{f}$ is non-decreasing and if we consider an informant $I$, since $\ps(I)\cup\ng(I)=\N$, $\simu_I$ is also unbounded.
In order to employ the delayability of $\beta$,
we define an operator $\Meth\colon \IfiSeq{(\Nttwo)} \to \IfiSeq{(\Nttwo)}$ such that for every $f \in \IfiSeq{(\Nttwo)}$ in form of $\Meth(f)$ we obtain a canonically sound version of $f$.
$\Meth(f)$ is defined on all $t < \simu_f(|f|)$ in case $f$ is finite and on every $t \in \N$ otherwise by
$$\Meth(f)(t):=\begin{cases}
(t,0), &\text{if } (t,0) \in \ran(f);\\
(t,1), &\text{otherwise.}
\end{cases}$$
Intuitively, in $\Meth(f)$ we sortedly and without repetitions sum up all information contained in $f$ up to the largest initial segment of $\N$, $f$ without interruption informs us about.
For a finite sequence $\sigma$ the canonical version $\Sigma(\sigma)$ has length $\simu_\sigma(|\sigma|)$.
Now consider the learner $M'$ defined by
$$M'(\sigma)=M(\Meth(\sigma)).$$
Since $I := \Meth(I')$ is a canonical informant for $L$, we have $\beta(M,I)$.
Moreover, for all $t \in \N$ holds $\ps(I[\simu_{I'}(t)]) \subseteq \ps(I'[t])$ and $\ng(I[\simu_{I'}(t)]) \subseteq \ng(I'[t])$ by the definitions of $\simu_{I'}$ and of $I$ using $\Meth$.
Finally, $$M'(I'[t])=M(\Meth(I'[t]))=M(\Meth(I')[\simu_{I'}(t)])=M(I[\simu_{I'}(t)])$$ and the delayability of $\beta$ yields $\beta(M',I')$.
\end{proof}

Therefore, while considering delayable learning from informant, looking only at canonical informant already yields the full picture also for set-driven learners.
Clearly, the picture is also the same for so-called \emph{partially set-driven learners} that base their hypotheses only on the set and the number of samples.


\bigskip
The next proposition answers the arising question, whether Lemma \ref{lem:canInfSdDelayableEx} also holds, when requiring the non-delayable learning restriction of consistency, negatively.

\medskip
$H$ denotes the halting problem.

\begin{proposition}
\label{prop:canInfConsEx}
For $\CalL := \{ 2 H \cup 2(H\cup\{x\})+1 \mid x \in \N \}$
holds
\begin{equation*}
\CalL \in [\totalCp\Inf_\can\Cons\Conv\SDec\SMon\Ex]\setminus[\Inf\Cons\Ex].
\end{equation*}
Particularly, $[\Inf\Cons\Ex] \subsetneq [\Inf_{\can}\Cons\Ex]$.
\end{proposition}
\begin{proof}
Let $p:\N\to\N$ be computable such that $W_{p(x)}=2 H \cup 2(H\cup\{x\})+1$ for every $x\in\N$ and let $h$ be an index for $2H\cup2H+1$. Consider the total learner $M$ defined by
\begin{equation*}
M(\sigma)=\begin{cases}
p(x), &\text{if } x \text{ with } 2x\in\ng(\sigma) \text{ and } 2x+1\in\ps(\sigma) \text{ exists}; \\
h, &\text{otherwise}
\end{cases}
\end{equation*}
for every $\sigma \in \Seq{(\Nttwo)}$.
Clearly, $M$ conservatively, strongly decisively and strongly monotonically $\Lim$-learns $\CalL$ from informant and on canonical informant for languages in $\CalL$ it is consistent.

\smallskip
Now, assume there is a learner $M$ such that $\CalL \in \Inf\Cons\Lim(M)$.
By Lemma~\ref{lem:lockingsequence} there is a locking sequence $\sigma$ for $2H\cup2H+1$.
By s-m-n there is a computable function 
\begin{align*}
\chi(x) = 
\begin{cases}
1, &\text{if } M(\sigma)= M(\sigma \concat (2x+1,1));\\
0, &\text{otherwise}.
\end{cases}
\end{align*}
By the consistency of $M$ on $\CalL$, we immediately obtain that $\chi$ is the characteristic function for $H$, a contradiction.
\end{proof}

Note, that there must not be an indexable family witnessing the difference stated in the previous proposition, since every indexable family is consistently and conservatively $\Lim$-learnable by enumeration.

\medskip
Further, \emph{request informant for $M$ and $L$} are introduced in \cite{Gol:j:67}. As the name already suggests, there is an interaction between the learner and the informant in the sense that the learner decides, about which natural number the informant should inform it next. His observation $[\Inf\Lim]=[\Inf_\can\Lim]=[\Inf_{\mathrm{req}}\Lim]$ seems to hold true when facing arbitrary delayable learning success criteria, but fails in the context of the non-delayable learning restriction of consistency.

\bigskip
Since 
$\CalL$ in Proposition~\ref{prop:canInfConsEx} lies in $[\Inf_\can\Lim]$, which by Lemma~\ref{lem:canInfSdDelayableEx} equals $[\Inf\Lim]$, we gain that for learning from informant consistent $\Lim$-learning is weaker than $\Lim$-learning, i.e., 
$[\Inf\Cons\Lim]\subsetneq[\Inf\Lim].$

\bigskip
We now show that, as observed for learning from text in \cite{Jai-Osh-Roy-Sha:b:99:stl2}, a consistent behavior regardless learning success cannot be assumed in general, when learning from informant.

\begin{proposition}
\label{TotalCons}
For $\CalL := \{\,\N,H\,\}$ holds
$$\CalL \in [\totalCp\Inf\Cons\Conv\SDec\Lim]\setminus[\Cons\Inf\Lim].$$
In particular, $[\Cons\Inf\Lim] \subsetneq [\Inf\Cons\Lim]$.
\end{proposition}
\begin{proof}
Fix an index $h$ for $H$ and an index $p$ for $\N$.
The total learner $M$ with 
$$M(\sigma)=\begin{cases}
p, &\text{if } \ng(\sigma)=\varnothing; \\
h, &\text{otherwise}
\end{cases}$$
for every $\sigma \in \Seq{(\Nttwo)}$ clearly consistently, conservatively and strongly decisively $\Lim$-learns $\CalL$.

\smallskip
Aiming at the claimed proper inclusion, assume there is a consistent learner $M$ for $\CalL$ from informant.
Since $M$ learns $H$, by Lemma~\ref{lem:lockingsequence}, we gain a locking sequence $\sigma \in \Seq{(\Nttwo)}$ for $M$ on $H$, which means $\Comp(\sigma,H)$, $W_{M(\sigma)}=H$ and for all $\tau \in \Seq{(\Nttwo)}$ with $\Comp(\tau,H)$ holds $M(\sigma\concat\tau)\!\!\downarrow=M(\sigma)$.
By letting
$$\chi(x) := \begin{cases}
1, &\text{if } M(\sigma\concat (x,1))=M(\sigma); \\
0, &\text{otherwise}
\end{cases}$$
for all $x\in\N$, we can decide $H$ by the global consistency of $M$, a contradiction. 
\end{proof}

\subsection{Total Learners}

Similar to full-information learning from text we show that for delayable learning restrictions totality is not a restrictive assumption.
Basically, the total learner simulates the original learner on the longest initial segment of the input, on which the convergence of the original learner is already visible.

\begin{lemma}
\label{lem:TotalInfDelayableEx}
Let $\beta$ be a delayable learning success criterion. Then 
\begin{equation*}
        [\Inf\beta]=[\mathcal{R}\Inf\beta].
\end{equation*}
\end{lemma}
\begin{proof}
Let $\CalL \in [\Inf\beta]$ and $M$ be a learner witnessing this. Without loss of generality we may assume that $\varnothing\in\dom(M)$.
We define the total learner $M'$ by letting
$\simu_M: \Seq{(\Nttwo)} \to \N,$
$$\sigma \mapsto \sup\{ s \in \N \mid s \leq |\sigma| \text{ and } M \text{ halts on } \sigma[s] \text{ after at most } |\sigma| \text{ steps} \}$$ and $$M'(\sigma):=M(\sigma[\simu_M(\sigma)]).$$
The convention $\sup(\varnothing)=0$ yields that $\simu_M$ is total and it is computable, since for $M$ only the first $|\sigma|$-many steps have to be evaluated on $\sigma$'s finitely many initial segments. One could also employ a Blum complexity measure here.
Hence, $M'$ is a total computable function.

In order to observe that $M'$ $\Inf\beta$-learns $\CalL$,
let $L \in \CalL$ and $I$ be an informant for $L$.
By letting $\simu(t) := \simu_M(I[t])$, we clearly obtain an unbounded non-decreasing function, hence $\simu \in \Simu$.
Moreover, for all $t\in\N$ from $\simu(t) \leq t$ immediately follows
\begin{align*}
\ps(I[\simu(t)]) &\subseteq \ps(I[t]), \;
\ng(I[\simu(t)]) \subseteq \ng(I[t]) \;\text{ as well as }\\
M'(I[t])&=M(I[\simu_M(I[t])])=M(I[\simu(t)]).
\end{align*}
By the delayability of $\beta$ and with $I'=I$, we finally obtain $\beta(M',I)$.
\end{proof}

\medskip
By the next proposition also for learning from informant requiring the learner to be total is a restrictive assumption for the non-delayable learning restriction of consistency. For learning from text this was observed in \cite{wiehagen1995learning} and generalized to $\delta$-delayed consistent learning from text in \cite{akama2008consistent}.

\begin{proposition}
\label{prop:TotalInfConsEx}
There is a \emph{collection of decidable languages} witnessing
$$[\totalCp\Inf\Cons\Lim] \subsetneq [\Inf\Cons\Lim].$$
\end{proposition}
\begin{proof}
Let $o$ be an index for $\varnothing$ and define for all $\sigma \in \Seq{(\Nttwo)}$ the learner $M$ by
$$
M(\sigma) := \begin{cases}
o, &\text{if } \ps(\sigma)=\varnothing; \\
\varphi_{\max(\ps(\sigma))}(\langle\sigma\rangle), &\text{otherwise.}
\end{cases}
$$
We argue that $\CalL := \{\, L \subseteq \N \mid L \text{ is decidable and } L\in\Inf\Cons\Lim(M) \,\}$ is not consistently learnable by a total learner from informant.
Assume towards a contradiction $M'$ is such a learner.
For a sequence $\sigma$ of natural numbers we denote by $\overline{\sigma}$ the corresponding canonical finite informant sequence, ending with the highest value $\sigma$ takes.
Further, for a natural number $x$ we denote by $\mathrm{seq}(x)$ the unique element of $\Seq{(\Nttwo)}$ with $\langle\mathrm{seq}(x)\rangle=x$.
Then by 1-1 ORT there are $e, z \in \N$ and functions $a, b: \Seq{\N} \to \N$, such that
\begin{equation}
\label{eq:abincreasing}
\forall \sigma, \tau \in \Seq{\N} \: (\, \sigma \prinseg \tau \Rightarrow \max\{a(\sigma),b(\sigma)\} < \min\{a(\tau),b(\tau)\}\,),
\end{equation}
with the property that for all $\sigma \in \Seq{\N}$ and all $i \in \N$
\begin{align}
\sigma_0 &= \varnothing; \nonumber \\
\sigma_{i+1} &= \sigma_i\,\concat \begin{cases}
a(\sigma_i), &\text{if } M'(\overline{\sigma_i\:\!\concat a(\sigma_i)})\neq M'(\overline{\sigma_i}); \\
b(\sigma_i), &\text{otherwise;}
\end{cases} \label{eq:MCNewLearner} \\
W_e &= \bigcup_{i \in \N} \ps(\overline{\sigma_i}); \nonumber \\
\varphi_z(y) &= \begin{cases}
1, &\text{if } y \in \ps(\overline{\sigma_y});\\
0, &\text{otherwise;}
\end{cases} \nonumber \\
\varphi_{a(\sigma)}(x)&=\begin{cases}
e, & \text{if } M'(\overline{\sigma\concat a(\sigma)})\neq  M'(\overline{\sigma}) \text{ and } \\
& \forall y \in \ps(\mathrm{seq}(x)) \: \varphi_z(y)=1 \:\wedge \\
& \forall y \in \ng(\mathrm{seq}(x)) \: \varphi_z(y)=0; \\
\ind(\ps(\mathrm{seq}(x))), & \text{otherwise;}
\end{cases} \nonumber \\
\varphi_{b(\sigma)}(x) &= \begin{cases}
e, & \text{if } \forall y \in \ps(\mathrm{seq}(x)) \: \varphi_z(y)=1 \:\wedge \\
& \forall y \in \ng(\mathrm{seq}(x)) \: \varphi_z(y)=0; \\
\ind(\ps(\mathrm{seq}(x))), & \text{otherwise;}
\end{cases} \nonumber
\end{align}
The operator $\Theta$ as stated in 1-1 $\ORT$ on page~\pageref{11ORT} is implicit in the equalities.
Further, $h$ is also implicitly given by $h(0)=e$, $h(1)=z$ and $a,b$ defined on all remaining even and odd numbers, respectively.
To be formally correct, the functions $a$ and $b$ rely on a computable encoding function with computable inverse mapping sequences $\sigma \in \N^{<\omega}$ to natural numbers and vice versa.

Let us now observe why the existence of such $e,z,a,b$ is contradictory.
Note that $\varphi_z$ witnesses $W_e$'s decidability by \eqref{eq:abincreasing} and with this whether $\varphi_{a(\sigma)}$ and $\varphi_{b(\sigma)}$ output $e$ or stick to $p$ depends on $\Comp(\mathrm{seq}(x),W_e)$.
Clearly, we have $W_e \in \CalL$ and thus $M'$ also $\Inf\Cons\Lim$-learns $W_e$.
By the $\Lim$-convergence there are $e', j \in \N$, where $j$ is minimal, such that $W_{e'}=W_e$ and for all $i \geq j$ we have $M'(\overline{\sigma_i})=e'$ and hence $M'(\overline{\sigma_i\:\!\concat a(\sigma_i)})= M'(\overline{\sigma_i})$ by \eqref{eq:MCNewLearner}.

We now argue that $L := \ps(\overline{\sigma_j}) \cup \{a(\sigma_j)\} \in \CalL$.
Let $I$ be an informant for $L$ and $t \in \N$.
By \eqref{eq:MCNewLearner} we observe that $M$ is consistent on $I$ as
$$M(I[t])=\varphi_{\max(\ps(I[t]))}(\langle I[t] \rangle) = \begin{cases}
e, &\text{if } \Comp(I[t],W_e); \\
\ind(\ps(I[t])), & \text{otherwise.}
\end{cases}$$
Further, by the choice of $j$ as well as \eqref{eq:abincreasing} and \eqref{eq:MCNewLearner} we have
\begin{equation}
\label{eq:asigmaj}
a(\sigma_j) \notin W_e = W_{e'},
\end{equation}
and with this $W_{M(I[t])}=L$, if $\ps(I[t])=L$.

On the other hand $M'$ does not consistently learn $L$ as by the choice of $j$ we obtain $M'(\overline{\sigma_j\!\,\concat a(\sigma_j)})=M'(\overline{\sigma_j})=e'$ and $\neg \Comp(\overline{\sigma_j\!\,\concat a(\sigma_j)},W_{e'})$ by \eqref{eq:asigmaj}, a contradiction.
\end{proof}

\section{Relations between Delayable Learning Success Criteria}
\label{sec:map_delayable}
In order to reveal the relations between the delayable learning restrictions in $\Lim$-learning from informant, we provide a regularity property of learners, called \emph{syntactic decisiveness}, for $\Lim$-learning in Lemma~\ref{lem:InfSynDecEx}.

Most importantly, in Proposition~\ref{GWlogConvSdec} we acquire that conservativeness and strongly decisiveness do not restrict informant learning. After this, Propositions~\ref{InfCautNotMonEx} and~\ref{InfMonNotCautEx} provide that cautious and monotonic learning are incomparable, implying that both these learning settings are strictly stronger than strongly monotonic learning and strictly weaker than unrestricted learning. The overall picture is summarized in Figure~\ref{MapInfDelayableEx} and stated in Theorem~\ref{InfExMap}.

\subsection{Syntactically Decisive Learning}

A further beneficial property, requiring a learner never to \emph{syntactically} return to an abandoned hypothesis, is supplied.

\begin{definition}[\cite{kotzing2016map}]
\label{def:SynDec}
Let $M$ be a learner, $L$ a language and $I$ an informant for $L$.
We write
\begin{enumerate}
\item[] $\SynDec(M,I)$, if $M$ is \emph{syntactically decisive} on $I$, i.e.,
$$\forall r,s,t \colon (r \leq s \leq t \wedge h_r = h_t) \Rightarrow h_r = h_s.$$
\end{enumerate}
\end{definition}

The following easy observation shows that this variant of decisiveness can always be assumed in the setting of $\Lim$-learning from informant. This is employed in the proof of our essential Proposition \ref{GWlogConvSdec}, showing that conservativeness and strong decisiveness do not restrict $\Lim$-learning from informant.

\begin{lemma}
\label{lem:InfSynDecEx}
We have $[\Inf\Lim]=[\SynDec\Inf\Lim]$.
\end{lemma}
\begin{proof}
Since obviously $[\SynDec\Inf\Lim]\subseteq[\Inf\Lim]$,
it suffices to show that every $\Inf\Lim$-learnable collection of languages is also $\SynDec\Lim$-learnable from informant.
For, let $\CalL \in [\Inf\Lim]$ and $M$ witnessing this. In the definition of the learner $M'$, we make use of a one-one computable padding function $\pad: \N \!\times\! \N \to \N$ such that $W_p = \dom(\varphi_p) = \dom(\varphi_{\pad(p,x)}) = W_{\pad(p,x)}$ for all $p,x \in \N$.
Now, consider $M'$ defined by
$$M'(\sigma):= \begin{cases}
\pad(M(\sigma),|\sigma|), & \text{if } M(\sigma^-)\neq M(\sigma); \\
M'(\sigma), & \text{otherwise.}
\end{cases}
$$
$M'$ behaves almost like $M$ with the crucial difference, that whenever $M$ performs a mind change, $M'$ semantically guesses the same language as $M$ did, but syntactically its hypothesis is different from all former ones. 
The padding function's defining property and the assumption that $M$ $\Inf\Lim$-learns $\CalL$ immediately yield the $\SynDec\Inf\Lim$-learnability of $\CalL$ by $M'$.
\end{proof}

Note that $\SDec$ implies $\SynDec$, which is again a delayable learning restriction. 
Thus by Lemma~\ref{lem:canInfSdDelayableEx}, in the proof of Lemma~\ref{lem:InfSynDecEx} we could have also restricted our attention to canonical informant.
It is further easy to see that Lemma~\ref{lem:InfSynDecEx} also holds for all other convergence criteria introduced and the simulation does not destroy any of the learning restrictions introduced in Definition~\ref{def:learningrestrictions}.

\subsection{Conservative and Strongly Decisive Learning}

The following proof for $\Conv\SDec\Lim$-learning being equivalent to $\Lim$-learning from informant builds on the normal forms of canonical presentations and totality provided in Section~\ref{sec:WlogTotalMethodical} as well as the regularity property introduced in the last subsection.

\begin{proposition}
\label{GWlogConvSdec}
We have $[\Inf\Lim]=[\Inf\Conv\SDec\Lim]$.
\end{proposition}

\begin{proof}
Obviously $[\Inf\Lim]\supseteq[\Inf\Conv\SDec\Lim]$ and
by the Lemmas \ref{lem:canInfSdDelayableEx}, \ref{lem:TotalInfDelayableEx} and \ref{lem:InfSynDecEx} it suffices to show
$[\totalCp\SynDec\Inf\Lim]\subseteq[\Inf_\can\Conv\SDec\Lim]$.

In the following for every set $X$ and $t \in \N$, let $X[t]$ denote the canonical informant sequence of the first $t$ elements of $\N$.

\smallskip
Now, let $\CalL \in [\totalCp\SynDec\Inf\Lim]$ and $M$ a learner witnessing this. In particular, $M$ is total and on informant for languages in $\CalL$ we have that $M$ never returns to a withdrawn hypothesis.
We want to define a learner $M'$ which mimics the behavior of $M$, but modified such that, if $\sigma$ is a locking sequence, then the hypothesis of $M'$ codes the same language as the guess of $M$.
However, if $\sigma$ is not a locking sequence, then the language guessed by $M'$ should not include data that $M$ changes its mind on in the future.
Thus, carefully in form of a recursively defined $\subseteq$-increasing sequence $(A^t_\sigma)_{t\in\N}$ in the guess of $M'$ we only include the elements of the hypothesis of $M$ that do not cause a mind change of $M$ when looking more and more computation steps ahead.
The following formal definitions make sure, this can be done in a computable way.

\smallskip
For every $\sigma \in \Seq{(\Nttwo)}$, $t \in\N$ with $t\geq |\sigma|$ and $D \subseteq W^t_{M(\sigma)}$, we let
$$
\mathbb{r}_\sigma^t(D) = \min \set{|\sigma|\leq r \leq t}{D \subseteq W^r_{M(\sigma)}}.
$$

Moreover, we define\footnote{We suppose $\inf(\emptyset) = \infty$ for convenience.}
\begin{align*}
\mathcal{X}_\sigma^t(D) = \{\,X \subseteq W^t_{M(\sigma)} \mid\: &\max(X) < \inf(W^t_{M(\sigma)} \setminus X) \text{, } D \subsetneq X \text{ and }\\& M(\sigma) = M(\,W^t_{M(\sigma)}[\,\mathbb{r}_\sigma^t(X)+1\,]\,)\,\}.
\end{align*}
In the following we abbreviate $X \subseteq W^t_{M(\sigma)}$ and $\max(X) < \inf(W^t_{M(\sigma)} \setminus X)$ by $X \inseg W^t_{M(\sigma)}$ and say that $X$ is an initial subset of $W^t_{M(\sigma)}$.

\smallskip
Aiming at providing suitable hypotheses $p(\sigma)$ for the conservative strongly decisive learner $M'$, given $\sigma$, we carefully enumerate more and more elements included in $W_{M(\sigma)}$.
We are going to start with the positive information provided by $\sigma$. Having obtained $A^t_{\sigma}$ with $\mathcal{X}^t_\sigma(A^t_\sigma)$ we have a set at hand that contains all initial subsets $X$ of $W^t_{M(\sigma)}$ strictly incorporating $A^t_\sigma$, for which $M$ does not differentiate between $\sigma$ and the appropriate initial segment $W^t_{M(\sigma)}[\,\mathbb{r}_\sigma^t(X)+1\,]$ of the canonical informant of $M$'s guess on $\sigma$.
Thus $\mathcal{X}^t_\sigma(A^t_\sigma)$ contains our candidate sets for extending $A^t_\sigma$.
The length $\mathbb{r}_\sigma^t(X)+1$ of the initial segment is minimal such that $X$ is a subset of $W^{\mathbb{r}_\sigma^t(X)}_{M(\sigma)}$ and at least $|\sigma|$ to assure $\Lim$-convergence of the new learner.

For an arbitrary $\sigma \in \Seq{(\Nttwo)}$ this reads as follows
\begin{align*}
A^0_\sigma & = \ps(\sigma);\\
\forall t \in \natnum: A^{t+1}_\sigma & = 
\begin{cases}
W_{M(\sigma)}^t,														
& \text{if } \ng(\sigma) \cap A_{\sigma}^t \neq \emptyset;\\
\max_{\subseteq} \mathcal{X}_\sigma^t( A^t_\sigma),	& \text{else if } \mathcal{X}_\sigma^t( A^t_\sigma ) \neq \emptyset;\\
A^t_\sigma,																& \mbox{otherwise.}
\end{cases}
\end{align*}
Furthermore, using s-m-n, we define $p: \Seq{(\Nttwo)} \to \N$ as a one-one function, such that for all $\sigma\in \Seq{(\Nttwo)}$
\begin{equation} \label{eq:DefHypConvLearner}
W_{p(\sigma)} = \bigcup_{t \in \natnum} A^t_\sigma.
\end{equation}
In the following, for all $\tau \in \Seq{(\Nttwo)}$ we denote by $\tau'$ the largest initial segment of $\tau$ for which $M'(\tau')=M'(\tau)$, i.e., the last time $M'$ performed a mind change.
Finally, we define our new learner $M'$ by
$$
M'(\sigma) =
\begin{cases}
p(\sigma), 		&\text{if } |\sigma|=0;\\
p(\sigma), 		&\text{else if } M((\sigma^-)')\neq M(\sigma) \wedge \neg \Comp(\sigma, A^{|\sigma|}_{(\sigma^-)'});\\
M'(\sigma^-),	&\text{otherwise.}
\end{cases}
$$
That is, $M'$ follows the mind changes of $M$ once a suitably inconsistent hypothesis has been seen. All hypotheses of $M$ are poisoned in a way to ensure that we can decide inconsistency.

\medskip
Let us first observe that $M'$ $\Lim$-learns every $L \in \Inf\Lim(M)$ from informant.
For, let $t_0$ be minimal such that, for all $t \geq t_0$, $M(L[t]) = M(L[t_0])$.
Thus, $e := M(L[t_0])$ is a correct hypothesis for $L$.

\smallskip
If $M'$ does not make a mind change in or after $t_0$, then $M'$ converged already before that mind change of $M$. Thus, let $s_0 < t_0$ be minimal such that for all $t \geq s_0$, $e' := M'(L[s_0]) = M'(L[t])$. 
As $p$ is one-one and $M$ learns syntactically decisive, we have $M(L[s_0])\neq M(L[t])$ for all $t \geq t_0$.
From $(L[t-1])'=L[s_0]$ and the definition of $M'$ we get $\Comp(L[t], A^{t}_{L[s_0]})$ for all $t\geq t_0$.
Thus, $W_{e'} = L$, because the final hypothesis $W_{e'}$ of $M'$ contains all elements of $L$ and no other by Equation~\eqref{eq:DefHypConvLearner}.

\smallskip
In case $M'$ makes a mind change in or after $t_0$, let $t_1 \geq t_0$ be the time of that mind change.
As $M$ does not perform mind changes after $t_0$, the learner $M'$ cannot make further mind changes and therefore converges to $e' := p(L[t_1])$.
By construction we have $A^t_{L[t_1]} \subseteq W_{e} = L$ for all $t\in\N$ and with it $W_{e'} \subseteq L$ by Equation~\ref{eq:DefHypConvLearner}.
Towards a contradiction, suppose $W_{e'} \subsetneq L$ and let $x \in L \setminus W_{e'}$ be minimal.
By letting $s_0$ such that $\ps(L[x])\subseteq A^{s_0}_{L[t_1]}$ and $x\in W^{s_0}_{e}$, every initial subset of $W^{s_0}_{e}$ extending $A^{s_0}_{L[t_1]}$ would necessarily contain $x$.
Therefore we have $A^s_{L[t_1]}=A^{s_0}_{L[t_1]}$ and $\mathcal{X}_{L[t_1]}^s(A^{s_0}_{L[t_1]})=\varnothing$ for all $s \geq s_0$.
We obtain the $\Lim$-convergence of $M'$ by constructing $s_2 \geq s_0$ with $\mathcal{X}_{L[t_1]}^{s_2}(A^{s_0}_{L[t_1]})\neq\varnothing$.
For this, let $y := \max (A^{s_0}_{L[t_1]} \cup \{x\})$ which implies $A^{s_0}_{L[t_1]} \subsetneq \ps(L[y+1])$.
Moreover, let $s_1 \geq t_1$ be large enough such that $L[y+1] = W_{e}^{s_1}[y+1]$.
Thus, by letting $r := \mathbb{r}^{s_1}_{L[t_1]}(\ps(L[y+1]))+1$
we gain $r=\mathbb{r}^{s}_{L[t_1]}(\ps(L[y+1]))+1$ for all $s \geq s_1$, where the latter denotes the time window considered in the third requirement for $\ps(L[y+1]) \in \mathcal{X}_{L[t_1]}^{s}(A^{s_0}_{L[t_1]})$.
Furthermore, let $s_2 \geq s_1$ with $L[r] = W^{s_2}_{e}[r]$.
By the definition of $r$ we have $r > t_1 \geq t_0$ and gain
\begin{align*}
\ps(L[y+1]) &\inseg W^{s_2}_{e}, \quad A^{s_2}_{L[t_1]} = A^{s_0}_{L[t_1]} \subsetneq \ps(L[y+1]) \quad \text{ and } \\ M(L[t_1]) &= e = M(L[r]) = M(\,W^{s_2}_{e}[r]\,),
\end{align*}
for short $\ps(L[y+1]) \in \mathcal{X}_{L[t_1]}^{s_2}(A^{s_0}_{L[t_1]})$, implying $\mathcal{X}_{L[t_1]}^{s_2}(A^{s_0}_{L[t_1]})\neq\varnothing$.

\medskip
Now we come to prove that $M'$ is conservative on every $L \in \Inf\Lim(M)$.
For, let $t$ be such that $M'(L[t]) \neq M'(L[t+1])$.
Let $e' := M'(L[t])$ and let $t' \leq t$ be minimal such that $M'(L[t']) = e'$.
From the mind change of $M'$ we get $\neg \Comp(L[t+1], A^{t+1}_{L[t']})$.
In case it holds $\ng(L[t+1])\cap A^{t+1}_{L[t']} \neq \varnothing$, since $A^{t+1}_{L[t']}\subseteq W_{e'}$, we would immediately observe $\neg \Comp(L[t+1],W_{e'})$. Therefore, we may assume $\ps(L[t+1])\setminus A^{t+1}_{L[t']}\neq \varnothing$.
Suppose, by way of contradiction, $W_{e'}$ is consistent with $L[t+1]$, i.e.,
$\ps(L[t+1]) \subseteq W_{e'}$ and $\ng(L[t+1]) \cap W_{e'} = \varnothing$. 
Then we have $\ng(L[t+1]) \cap A^s_{L[t']} = \varnothing$ for all $s \in \N$. 
Since $\ps(L[t+1]) \subseteq W_{e'}$, there is $t_0$ minimal such that
\begin{align}
L[t+1] = A^{t_0+1}_{L[t']}[t+1]. \label{eq:insegLA}
\end{align}
We have $\ng(L[t']) \cap A^{t_0}_{L[t']} = \varnothing$ as otherwise $\neg \Comp(L[t+1],W_{e'})$.
Because $t_0$ was minimal, we have $A^{t_0}_{L[t']} \subsetneq A^{t_0+1}_{L[t']}$ and with this $A^{t_0+1}_{L[t']}\in \mathcal{X}_{L[t']}^{t_0}(A^{t_0}_{L[t']})$ by the definition of $A^{t_0+1}_{L[t']}$.
In particular, this tells us
\begin{align} \label{eq:insegAW}
A^{t_0+1}_{L[t']} &\inseg W^{t_0}_{M(L[t'])} \quad \text{ and } \\ \label{eq:noMindChange}
M(L[t']) &= M(\,W^{t_0}_{M(L[t'])}[\,\mathbb{r}_{L[t']}^{t_0}(A^{t_0+1}_{L[t']})+1\,]\,).
\end{align}
and therefore with
\begin{align*}
L[t'] \inseg L[t+1] \stackrel{\eqref{eq:insegLA}}{\inseg} A^{t_0+1}_{L[t']} \stackrel{\eqref{eq:insegAW}}{\inseg} W^{t_0}_{M(L[t'])}[\,\mathbb{r}_{L[t']}^{t_0}(A^{t_0+1}_{L[t']})+1\,]
\end{align*}
by Equation~\eqref{eq:noMindChange} and $M$'s syntactic decisivenes we get $M(L[t'])=M(L[t+1])$.
Therefore, $M'$ did not make a mind change in $t+1$, a contradiction.
\end{proof}

\subsection{Completing the Picture of Delayable Learning}

The next two propositions show that monotonic and cautious $\Lim$-learning are incomparable on the level of indexable families. With Proposition~\ref{GWlogConvSdec} this yields all relations between delayable $\Lim$-learning success criteria as stated in Theorem~\ref{InfExMap}.

\medskip
We extend the observation of \cite{STL1} for cautious learning to restrict learning power with the following result. The positive part has already been discussed in the example in the introduction.

\begin{proposition}
\label{InfMonNotCautEx}
For the \emph{indexable family} $\CalL := \{ \N\setminus X \mid X \subseteq \N \text{ finite} \}$ holds
\begin{equation*}
\CalL \in [\Inf\Mon\Lim]\setminus[\Inf\Caut\Bc].
\end{equation*}
Particularly, $[\Inf\Caut\Lim] \subsetneq [\Inf\Lim]$.
\end{proposition}

\begin{proof}
In order to approach $\CalL \notin [\Inf\Caut\Bc]$,
let $M$ be a $\Inf\Bc$-learner for $\CalL$ and $I_0$ the canonical informant for $\N$.
Moreover, let $t_0$ be such that $W_{M(I_0[t_0])}=\N$.
Let $I_1$ be the canonical informant for $L_1 := \N\setminus\{t_0\}$. Since $M$ learns $L_1$, there is $t_1 > t_0$ such that $W_{M(I_1[t_1])}=L_1$. We have $I_1[t_0]=I_0[t_0]$ and hence $M$ is not cautiously learning $L_1$ from $I_1$.

\medskip
We now show the $\Mon\Ex$-learnability.
By s-m-n there is a computable function $p: \N \to \N$ such that for all finite sets $X$ holds $W_{p(\langle X\rangle)}=\N\setminus X$, where $\langle X \rangle$ denotes a canonical code for $X$ as already employed in the proof of Proposition \ref{InfCautNotMonEx}.
We define the learner $M$ by letting for all $\sigma \in \Seq{\Nttwo}$
\begin{align*}
M(\sigma)&= p(\langle \ng(\sigma) \rangle).
\end{align*}

The corresponding intuition is that $M$ includes every natural number in its guess, not explicitly excluded by $\sigma$.
Clearly, $M$ learns $\CalL$ and behaves monotonically on $\CalL$, since for every $X \subseteq \N$ finite, every informant $I$ for $\N\setminus X$ and every $t \in \N$, we have $W_{M(I[t])}\supseteq \N\setminus X$ and therefore $W_{M(I[t])}\cap \N\setminus X = \N \setminus X$.
\end{proof}

This reproves $[\Inf\SMon\Lim] \subsetneq [\Inf\Mon\Lim]$ observed in \cite{lange1996monotonic} also on the level of indexable families.

\bigskip
In the next proposition the learner can even be assumed cautious on languages it does not identify. Thus, according to Definition~\ref{def:learningcriterium} we write this success independent property of the learner on the left side of the mode of presentation.

\begin{proposition}
\label{InfCautNotMonEx}
For the \emph{indexable family}
$$\CalL := \{ 2X \cup (2(\N\setminus X)+1) \mid X \subseteq \N \text{ finite or } X = \N \}$$
holds $\CalL \in [\Caut\Inf\Lim]\setminus[\Inf\Mon\Bc]$.

Particularly, $[\Inf\Mon\Lim] \subsetneq [\Inf\Lim]$.
\end{proposition}

\begin{proof}
We first show $\CalL \notin [\Inf\Mon\Bc]$.
Let $M$ be a $\Inf\Bc$-learner for $\CalL$.
Further, let $I_0$ be the canonical informant for
$L_0 := 2\N \in \CalL$.
Then there exists $t_0$ such that $W_{M(I_0[2t_0])}=2\N$. 
Moreover, consider the canonical informant $I_1$ for
$$L_1 := 2\{0,\ldots,t_0\}\cup (2(\N\setminus\{0,\ldots,t_0\})+1)
\; \in \CalL$$
and let $t_1>t_0$ such that $W_{M(I_1[2t_1])}=L_1$.
Similarly, we let $I_2$ be the canonical informant for 
$$L_2 := 2\{0,\ldots,t_0,t_1+1\}\cup (2(\N\setminus\{0,\ldots,t_0,t_1+1\})+1)
\; \in \CalL$$
and choose $t_2>t_1$ with $W_{M(I_2[2t_2])}=L_2$.
Since $2(t_1+1) \in (L_0 \cap L_2) \setminus L_1$ and by construction $I_2[2t_0]=I_0[2t_0]$ as well as $I_2[2t_1]=I_1[2t_1]$, we obtain
$$2(t_1+1) \in W_{M(I_2[2t_0])} \cap L_2 \quad \text{ and } \quad 2(t_1+1) \notin W_{M(I_2[2t_1])} \cap L_2$$ and therefore $M$ does not learn $L_2$ monotonically from $I_2$.

\medskip
Let us now adress $\CalL \in [\Caut\Inf\Lim]$.
Fix $p \in \N$ such that $W_p = 2\N$. Further, by s-m-n there is a computable function $q: \N \to \N$ with $W_{q(\langle X\rangle)} = X \cup (2\N \setminus X)+1$, where $\langle X \rangle$ stands for a canonical code of the finite set $X$.
We define the learner $M$ for all $\sigma \in \Seq{\Nttwo}$ by
\begin{align*}
M(\sigma) &= \begin{cases}
p, & \text{if } \ps(\sigma) \subseteq 2\N; \\
q(\langle\ps(\sigma)\cap2\N \rangle), & \text{otherwise.}
\end{cases}
\end{align*}
Intuitively, $M$ guesses $2\N$ as long as no odd number is known to be in the language $L$ to be learned.
If for sure $L \neq 2\N$, then $M$ assumes that all even numbers known to be in $L$ so far are the only even numbers therein.

It is easy to verify that $M$ is computable and by construction it learns $\CalL$.
For establishing the cautiousness, let $L$ be any language, $I$ an informant for $L$ and $s \leq t$.
Furthermore, assume $W_{M(I[s])} \neq W_{M(I[t])}$.
In case $\ps(I[s]) \not \subseteq 2\N$, we have $x \in (\ps(I[t])\cap 2\N)$ with $x \notin (\ps(I[s])\cap 2\N)$ and therefore as desired
$W_{M(I[t])} \setminus W_{M(I[s])} \neq \varnothing$.
Then $\ps(I[s]) \subseteq 2\N$ implies $W_{M(I[s])}=2\N$ and thus again
$W_{M(I[t])} \setminus W_{M(I[s])} \neq \varnothing$.
\end{proof}

\smallskip
We sum up the preceding results in the next theorem and also represent them in Figure \ref{MapInfDelayableEx}.

\begin{figure}[h]
\begin{center}
\begin{minipage}{12cm}
\includegraphics[scale=1]{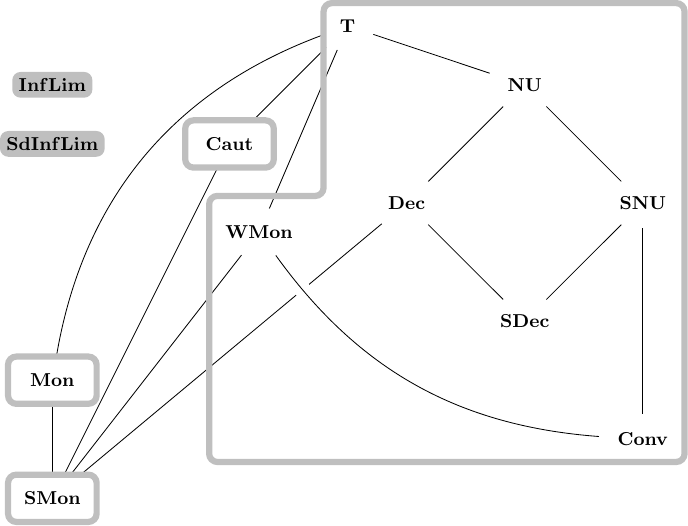}
\end{minipage}
\end{center}
\caption{Relations between delayable learning restrictions in full-information (explanatory) $\Lim$-learning of languages from informant.
The implications according to Lemma~\ref{lem:delayablebackbone} are represented as black lines from bottom to top. Two learning settings are equivalent if and only if they lie in the same grey outlined zone as stated in Theorem~\ref{InfExMap}.
}
\label{MapInfDelayableEx}
\end{figure}

\begin{theorem}
\label{InfExMap}
We have
\begin{enumerate}
\item $\forall\delta \in \{\Conv,\Dec,\SDec,\WMon,\NU,\SNU\}:$ $$[\Inf\delta\Lim]=[\Inf\Lim].$$ 
\item $[\Inf\Mon\Lim]\perp[\Inf\Caut\Lim]$.
\end{enumerate}
\begin{proof}
The first part is an immediate consequence of Proposition \ref{GWlogConvSdec} and so is the second part of the Propositions \ref{InfMonNotCautEx} and \ref{InfCautNotMonEx}.
\end{proof}
\end{theorem}

\section{Outperforming Learning from text}
\label{sec:infvstxt}
Already in \cite{Gol:j:67} it was observed that $[\Txt\Lim] \subsetneq [\Inf\Lim]$.
Later on in \cite{Lan-Zeu:c:93} the interdependencies when considering the different monotonicity learning restrictions were investigated.
For instance, they showed that there exists an indexable family $\CalL \in [\Inf\Mon\Lim]\setminus[\Txt\Lim]$ and in contrast that for indexable families $\Inf\SMon\Lim$-learnability implies $\Txt\Lim$-learnability.  
We show that this inclusion fails on the level of families of recursive languages even with all learning restrictions at hand.

%

\begin{proposition}
For the class of recursive languages
$$\CalL := \{ 2(L\cup\{x\})\cup 2L+1 \mid L \text{ is recursive} \wedge W_{\min(L)}=L \wedge x \geq \min(L) \}$$
holds
$\CalL \in [\Inf\Conv\SDec\SMon\Lim]\setminus[\Txt\Lim].$
\end{proposition}
\begin{proof}
Let $p_{m}$ denote an index for $2W_m\cup 2W_m+1$ and $p_{m,x}$ an index for $2(W_m\cup\{x\})\cup 2W_m+1$.
The learner $M$ will look for the minimum of the presented set and moreover try to detect the exception $x$, in case it exists. Thus, it checks for all $m$ such that $2m \in \ps(\sigma)$ or $2m+1 \in \ps(\sigma)$ whether for all $k<m$ holds $2k \in \ng(\sigma)$ or $2k+1 \in \ng(\sigma)$.
In case $m$ has this property relative to $\sigma$, we write $\min_L(m,\sigma)$ as $m$ might be the minimum of the language presented.
Further, $M$ tries to find $x$ such that $2x \in \ps(\sigma)$ and $2x+1 \in \ng(\sigma)$ and we abbreviate by $\mathrm{exc}_L(x,\sigma)$ that $x$ is such an exception.
Consider the learner $M$ for all $\sigma \in \Seq{(\Nttwo)}$ defined by
$$M(\sigma)=
\begin{cases}
\ind(\varnothing), &\text{if there is no $m$ with } \min_L(m,\sigma);\\
p_m, &\text{if } \min_L(m,\sigma) \text{ and there is no $x$ with } \mathrm{exc}_L(x,\sigma);\\
p_{m,x}, &\text{if } \min_L(m,\sigma) \text{ and $x$ is minimal with } \mathrm{exc}_L(x,\sigma). \\ 
\end{cases}
$$
Clearly, $M$ conservatively, strongly decisively and strongly monotonically $\Lim$-learns $\CalL$.

To observe $\CalL \notin [\Txt\Lim]$, assume there exists $M$ such that $\CalL \in \Txt\Lim(M)$.
By ORT there exists $e\in\N$ such that for all $i\in\N$
\begin{align*}
A_\sigma(i) &= \{\,k\in\N \mid M(\sigma) \neq M(\sigma\concat(2e + 4i)^k) \,\};\\
B_\sigma(i) &= \{\,k\in\N \mid M(\sigma) \neq M(\sigma\concat(2e + 4i + 2)^k) \,\}; \\
\sigma_0 &= (2e,2e+1);\\
\sigma_{i+1} &= 
\begin{cases}
\sigma_i, &\text{if } A_{\sigma_i}(i) = B_{\sigma_i}(i) = \varnothing \\
&\text{or } i>0 \wedge \sigma_{i-1}=\sigma_i;\\
\sigma_i \concat(2e+4i)^{\inf(A_{\sigma_i}(i))}{\concat}(2e + 4i + 1), &\text{if } A_{\sigma_i}(i)\neq\varnothing\\
&\wedge \inf(A_{\sigma_i}(i)) \leq \\
&\qquad\qquad\quad\inf(B_{\sigma_i}(i)));\\
\sigma_i \concat (2e+4i+2)^{\inf(B_{\sigma_i}(i))}{\concat}(2e + 4i + 3), &\text{if } B_{\sigma_i}(i)\neq\varnothing\\
&\wedge \inf(B_{\sigma_i}(i)) < \\
&\qquad\qquad\quad\inf(A_{\sigma_i}(i));
\end{cases}\\
W_{e} &= \bigcup_{i\in\N} \{n \mid 2n+1 \in \ran(\sigma_i) \}.
\end{align*}
Intuitively, the program on input $n$ successively computes $\sigma_i$ until it finds the minimal $x\geq 2n+1$ in its range; it halts if and only if $x$ is found and $x=2n+1$.

$W_e$ is recursive, because 
we can decide it along the construction of the $\sigma_i$. Thus, $2W_e \cup 2W_e+1 \in \mathcal{L}$.
If for some index $i$ holds $\sigma_{i+1}=\sigma_i$, then $M$ fails to learn $2(W_e\cup\{e+2i\}) \cup 2W_e+1$ or $2(W_e\cup\{e+2i + 1\}) \cup 2W_e+1$.
On the other hand, if there is no such $i$, by letting $T := \bigcup_{i\in\N} \sigma_i$ we obtain a text for $2W_e \cup 2W_e+1$, on which $M$ performs infinitely many mindchanges.
\end{proof}

\section{Anomalous Hierarchy and Vacillatory Duality}
\label{sec:hierarchies}
We compare the convergence criteria $\Lim^a_b$ from Definition~\ref{def:InfConvergence} for different parameters $a \in \N\cup\{\ast\}$ and $b \in \N_{>0}\cup\{\ast,\infty\}$.
The duality depending on whether $b=\infty$ for fixed $a$ follow from the Propositions~\ref{InfBcStrongerEx} and \ref{collapshierarchy}.

\subsection{Anomalous Hierarchy}

\medskip
Beneficial for analyzing the anomalous hierarchy for informant learning are results from function learning.
When learning collections of recursive functions, a text for the graph of the respective function $f$ is presented to the learner and it wants to infer a program code $p$ such that $\varphi_p$ is a good enough approximation to $f$.
More formally, $f=^a \varphi_p$ if and only if $|\{\, x \,|\, f(x) \neq \varphi_p(x) \,\}|\leq a$.
We denote the associated learning criteria in the form $[\Fn\Lim^a_b]$.

\medskip
By the next lemma, collections of functions separating two convergence criteria in the associated setting yield a separating collection for the respective convergence criteria, when learning languages from informant.

\medskip
In the following we make use of a computable bijection $\langle.\;\!,.\rangle: \N\times\N \to \N$ with its computable inverses $\pi_1, \pi_2: \N\to\N$ such that $x=\langle\pi_1(x),\pi_2(x)\rangle$ for all $x\in\N$.

\begin{lemma}
\label{FnResultsToInf}
For $f \in \totalCp$ let $L_f:=\{\,\langle x,f(x)\rangle \mid x\in\N\,\}$ denote the language encoding its graph.
Let $a \in \N\cup\{\ast\}$ and $b\in\N_{>0}\cup\{\ast,\infty\}$.
Then for every $\CalF \subseteq \totalCp$ we define $\CalL_\CalF = \{\,L_f \mid f \in \CalF\,\}$ and obtain
$$\CalF \in [\Fn\Lim^a_b] \quad \Leftrightarrow \quad \CalL_\CalF \in [\Inf\Lim^a_b].$$
\end{lemma}
\begin{proof}
Let $a,b$ and $\CalF$ be as stated.
First, assume there is a learner $M$ on function sequences such that $\CalF \in \Fn\Lim^a_b(M)$.
In order to define the learner $M'$ acting on informant sequences and returning $W$-indices, we employ the following procedure for obtaining a $W$-code $G(p)$ for $L_{\varphi_p}$, when given a $\varphi$-code $p$:
\begin{quote}
Given input $n$, interpreted as $\langle x,y\rangle$, let the program encoded by $p$ run on $x=\pi_1(n)$. If it halts and returns $y=\pi_2(n)$, then halt; otherwise loop.
\end{quote}
The learner $M'$ acts on $\sigma \in \Seq{(\Nttwo)}$ by
$$M'(\sigma) := G(M(\decode(\ps(\sigma)))),$$
where $\decode(\ps(\sigma))$ denotes the from $\sigma$ uniformly computable sequence $\tau$ with $\tau(i)=(\pi_1(n_i),\pi_2(n_i))$ for all $i<|\ps(\sigma)|=|\tau|$, where $(n_i)_{i<|\ps(\sigma)|}$ denotes the enumeration of $\ps(\sigma)$ according to $\sigma$.
By construction, $\CalL_\CalF \in \Inf\Lim^a_b(M')$ as $G$ preserves the number of anomalies.

For the other claimed direction let $M$ be a learner on informant sequences with $\CalL_\CalF \in \Inf\Lim^a_b(M)$.
As above we employ a computable function that for every $f\in\totalCp$ transforms a $W$-index $p$ for $L_f$ into a $\varphi$-index $H(p)$ such that $\varphi_{H(p)}=f$. Thereby, we interpret each natural number $i$ as $\langle \langle u,v\rangle,t\rangle$ and check whether $\varphi_p$ halts on $\langle u,v\rangle$ in at most $t$ steps of computation. If so, we check whether $u$ is the argument $x$ we want to compute $f(x)$ for and in case the answer is yes, we return $v$.
\begin{quote}
Given input $x$, for $i=0$ till $\infty$ do the following: If $\Phi_p(\pi_1(i))\leq \pi_2(i)$ and $\pi_1(\pi_1(i))=x$, then return $\pi_2(\pi_1(i))$; otherwise increment $i$.
\end{quote}
Before defining $M'$, we argue that $H$ preserves the number of anomalies.
Let $x,p\in\N$ be such that $f(x)\neq\varphi_{H(p)}(x)$.
Then $\langle x,\varphi_{H(p)}(x)\rangle\notin L_f$.
On the other hand, by the definition of $H$ we have $\Phi_p(\langle x,\varphi_{H(p)}(x)\rangle)\!\!\downarrow$ and therefore $\langle x,\varphi_{H(p)}(x)\rangle\in W_p\setminus L_f$.

We define the learner $M'$ on $\sigma \in \Seq{(\N\!\times\!\N)}$ by
$$M'(\sigma):=H(M(\hat\sigma)),$$
where we transform $\sigma = ((x_0,f(x_0)),\ldots,(x_{|\sigma|-1},f(x_{|\sigma|-1})))$ into an informant sequence $\hat\sigma$ of length $\widehat{|\sigma|}:=\max\{j\mid \forall i<j \: \pi_1(i)< |\sigma| \}$ by letting
$$
\hat\sigma(i):=
\begin{cases}
(\langle x_{\pi_1(i)},\pi_2(i)\rangle,1) &\text{if } \sigma(\pi_1(i))=(x_{\pi_1(i)},\pi_2(i)) \\
(\langle x_{\pi_1(i)},\pi_2(i)\rangle,0) &\text{otherwise}
\end{cases}
$$
for all $i<\widehat{|\sigma|}$. Note that for every $f\in\totalCp$ and every $T\in\Txt(f)$ by letting $I_T := \bigcup_{j\in\N} \widehat{T[j]}$, we obtain an informant for $L_f$.
We show $(\langle x,f(x)\rangle,1) \in I_T$ for every $x\in\N$ and leave the other details to the reader.
Let $x\in\N$ and $i$ minimal, such that $(x,f(x))\in \ran(T[i])$, i.e., $x_{i-1}=x$.
Further, let $j$ be such that $i\leq\hat \jmath$.
Then clearly
$$I_T(\langle i-1, f(x)\rangle)=\widehat{T[j]}(\langle i-1, f(x)\rangle)=(\langle x, f(x)\rangle,1).$$
In a nutshell, $\CalF \in \Fn\Lim^a_b(M')$ as $H$ preserves the number of anomalies.
\end{proof}

With this we obtain a hierarchy of learning restrictions.

\begin{proposition}
\label{hierarchyLimab}
Let $b \in \{1,\infty\}$. Then
\begin{enumerate}[label=(\roman*)]
\item for all $a \in \N$ holds $[\Inf\Lim^a_b]\subsetneq[\Inf\Lim^{a+1}_b]$,
\item $\bigcup_{a \in \N} [\Inf\Lim^a_b] \subsetneq [\Inf\Lim^\ast_b]$,
\item \label{ExastBc} $[\Inf\Lim^\ast] \subsetneq [\Inf\Bc]$.
\end{enumerate}
\end{proposition}
\begin{proof}
By Lemma \ref{FnResultsToInf} this results transfer from the corresponding observations for function learning in \cite{Bar:j:74:two-thrm} and \cite{case1983comparison}.
\end{proof}

In particular, we have
\begin{align*}
[\Inf\Lim] &\subsetneq \ldots \subsetneq [\Inf\Lim^a] \subsetneq [\Inf\Lim^{a+1}] \subsetneq \ldots \\
&\subsetneq \bigcup_{a \in \N} [\Inf\Lim^a] \subsetneq  [\Inf\Lim^\ast] \\
&\subsetneq [\Inf\Bc] \subsetneq \ldots \subsetneq [\Inf\Bc^a] \subsetneq [\Inf\Bc^{a+1}] \subsetneq \ldots \\
&\subsetneq \bigcup_{a \in \N} [\Inf\Bc^a] \subsetneq  [\Inf\Lim^\ast_\infty].
\end{align*}

Lemma~\ref{FnResultsToInf} obviously also holds when considering $\Txt\Lim^a_b$-learning languages, where the construction of the text sequence from the informant sequence is folklore. This reproves the results in \cite{Cas-Lyn:c:82}.

\subsection{Duality of the Vacillatory Hierarchy}

In Proposition \ref{hierarchyLimab} we already observed a hierarchy, when varying the number of anomalies and will now show that allowing the learner to vacillate between finitely many correct hypothesis in the limit does not give more learning power. On the contrary, only requiring semantic convergence, i.e., allowing infinitely many correct hypotheses in the limit, does allow to learn more collections of languages even with an arbitrary semantic learning restriction at hand.
This contrasts the results in language learning from text in \cite{case1999power}, observing for every $a \in \N\cup\{\ast\}$ a hierarchy
\begin{align*}
[\Txt\Lim^a] &\subsetneq \ldots \subsetneq [\Txt\Lim^a_b] \subsetneq [\Txt\Lim^a_{b+1}] \subsetneq \ldots \\
&\subsetneq \bigcup_{b \in \N_{>0}} [\Txt\Lim^a_b] \subsetneq  [\Txt\Lim^a_\ast] \subseteq [\Txt\Bc^a].
\end{align*}

\medskip
We separate $\Inf\Lim$- and $\Inf\Bc$-learning at the level of families of recursive languages, even when requiring the $\Bc$-learning sequence to meet all introduced delayable semantic learning restrictions.
As every indexable family of recursive languages is $\Lim$-learnable from informant by enumeration, the result is optimal.

\begin{proposition}
\label{InfBcStrongerEx}
For the collection of recursive languages
$$\CalL = \{ L \cup \{x\} \mid L \subseteq \N \text{ is recursive} \wedge W_{\min(L)}=L \wedge x \geq \min(L) \}$$
holds
$\CalL \in [\Inf\SMon\Bc]\setminus[\Inf\Lim].$
\end{proposition}
\begin{proof}
By Lemma~\ref{lem:canInfSdDelayableEx} it suffices to show $$\CalL \in [\Inf_\can\SMon\Bc]\setminus[\Inf_\can\Lim].$$
By s-m-n there are $p: \Seq{\Nttwo} \times \N \to \N$ and a learner $M$ such that for all $\sigma \in \Seq{\Nttwo}$ and $x \in \N$
\begin{align*}
W_{p(\sigma,x)} &= W_{\min(\ps(\sigma))}\cup\{x\} \text{ and} \\
M(\sigma) &=
\begin{cases}
o, &\text{if } \ps(\sigma)=\varnothing; \\
\min(\ps(\sigma)), &\text{else if } \ps(\sigma)\setminus W_{\min(\ps(\sigma))}^{|\sigma|} = \varnothing; \\
p(\sigma,x), &\text{else if } x=\min(\ps(\sigma)\setminus W_{\min(\ps(\sigma))}^{|\sigma|});
\end{cases} \\
\intertext{
where $o$ refers to the canonical index for the empty set.
Let $L\cup\{x\} \in \CalL$ with $L \subseteq \N$ recursive, $W_{\min(L)}=L$ and $x \geq \min(L)$
and let $I$ be the canonical informant for $L \cup \{x\}$.
Then for all $t > \min(L)$ we have $W_{\min(\ps(I[t]))}=W_{\min(L)}=L$.
Further, let $m$ be minimal such that $\{ y\in L \mid y<x\} \subseteq W_{\min(L)}^m$.
Since $x\geq \min(L)$ the construction yields for all $t \in \N$
}
W_{h_t}&=
\begin{cases}
\varnothing, &\text{if } t \leq \min(L); \\
L, &\text{else if } \min(L) \leq t < \max\{x+1,m\}; \\
L \cup \{x\}, &\text{otherwise.}
\end{cases}
\end{align*}
This can be easily verified, since in case $y \in L$ we have $L = L\cup\{y\}$ and establishes the $\Inf_\can\SMon\Bc$-learnability of $\CalL$ by $M$.

\medskip
In order to approach $\CalL \notin [\Inf_\can\Lim]$, assume to the contrary that there is a learner $M$ that $\Inf_\can\Lim$-learns $\CalL$.
By Lemma \ref{lem:TotalInfDelayableEx} $M$ can be assumed total.
We first define a recursive language $L$ with $W_{\min(L)}=L$ helpful for showing that not all of $\CalL$ is $\Inf_\can\Lim$-learned by $M$.
In order to do so, for every canonical $\sigma \in \Seq{\Nttwo}$ we define sets $A^0_\sigma, A^1_\sigma \subseteq \N$. For this let $I^0_\sigma$ stand for the canonical informant of $\ps(\sigma)$, whereas $I^1_\sigma$ denotes the canonical informant of $\ps(\sigma)\cup\{|\sigma|\}$.
In $A^0_\sigma$ we collect all $t > |\sigma|$ for which $M$'s hypothesis on $I^0_\sigma[t]$ is different from $M(\sigma)$. Similarly, in $A^1_\sigma$ we capture all $t  > |\sigma|$ such that $M$ on $I^1_\sigma[t]$ makes a guess different from $M(\sigma)$.
Formally, this reads as follows
\setcounter{footnote}{1}
\begin{align*}
A^0_\sigma &:= \{\, t \in \N \mid t > |\sigma| \wedge M(I^0_\sigma[t])\neq M(\sigma)\,\}, \\
A^1_\sigma &:= \{\, t \in \N \mid t > |\sigma| \wedge M(I^1_\sigma[t])\neq M(\sigma)\,\}.
\intertext{
As $\sigma$ is canonical, for every $t > |\sigma|$
}
I^0_\sigma[t] &= \sigma\concat (\,(|\sigma|,0),(|\sigma|\!+\!1,0),\ldots,(t-1,0)\,), \\
I^1_\sigma[t] &= \sigma\concat (\,(|\sigma|,1),(|\sigma|\!+\!1,0),\ldots,(t-1,0)\,).
\intertext{
By ORT there exists $p \in \N$ such that\footnotemark
}
\sigma_0 &= (\,(0,0), \ldots, (p-1,0), (p,1) \,),\\
\forall i \in \N \colon
\sigma_{i+1} &=
\begin{cases}
\sigma_i, & \text{if }A^0_{\sigma_i} = A^1_{\sigma_i} = \emptyset ;\\
I^0_{\sigma_i}[\min(A^0_{\sigma_i})], &\text{if } \inf(A^0_{\sigma_i}) \leq \inf(A^1_{\sigma_i});\\
I^1_{\sigma_i}[\min(A^1_{\sigma_i})], &\text{otherwise;}
\end{cases}\\
W_p &= \bigcup_{i \in \N} \ps(\sigma_i).
\intertext{
Intuitively, the program $p$ on input $x$ halts if $x=p$ or in the successive construction of the sequence $(\sigma_i)_{i \in \N}$ there is $j$ with $|\sigma_j|>x$ and $\sigma_j(x)=(x,1)$.
Hence, $p = \min(W_p)$ and $W_p$ is recursive, which immediately yields $L := W_p \in \CalL$.
Further, for every $i \in \N$ from $\sigma_i\neq\sigma_{i+1}$ follows $M(\sigma_i)\neq M(\sigma_{i+1})$.
Aiming at a contradiction, let $I$ be the canonical informant for $L$, which implies $\bigcup_{i\in\N}\sigma_i \inseg I$.
Since $M$ $\Lim$-learns $L$ and thus does not make infinitely many mind changes on $I$, there exists $i_0 \in \N$ such that for all $i \geq i_0$ we have $\sigma_i=\sigma_{i_0}$.
But then for all $t > |\sigma_{i_0}|$ holds
}
M(I^0_{\sigma_{i_0}}[t]) &= M(\sigma_{i_0}) = M(I^1_{\sigma_{i_0}}[t]),
\end{align*}
\footnotetext{Again we use the convention $\inf(\emptyset) = \infty$.}
thus $M$ does not learn at least one of $L = \ps(\sigma_{i_0})$ and $L \cup\{|\sigma_{i_0}|\}$ from their canonical informant. On the other hand both of them lie in $\CalL$ and therefore, $M$ had not existed in the beginning. 
\end{proof}

Since allowing infinitely many different correct hypotheses in the limit gives more learning power, the question arises, whether finitely many hypotheses already allow to learn more collections of languages. The following proposition shows that, as observed in \cite{Bar-Pod:c:73} and \cite{case1983comparison} for function learning, the hierarchy of vacillatory learning collapses when learning languages from informant.

\begin{proposition}
\label{collapshierarchy}
Let $a \in \N\cup\{\ast\}$. Then $[\Inf\Lim^a]=[\Inf\Lim^a_\ast]$.
\end{proposition}
\begin{proof}
Clearly, $[\Inf\Lim^a]\subseteq[\Inf\Lim^a_\ast]$.
For the other inclusion let $\CalL$ be in $[\Inf\Lim^a_\ast]$ and $M$ a learner witnessing this.
By Lemma~\ref{lem:TotalInfDelayableEx} we assume that $M$ is total.
In the construction of the $\Lim^a$-learner $M'$, we employ the recursive function
$\Xi: \Seq{(\Nttwo)}\times\N\to\N$, which given $\sigma \in \Seq{(\Nttwo)}$ and $p\in\N$ alters $p$ such that $W_{\Xi(\sigma,p)}^{|\sigma|}\cap\ng(\sigma)=\varnothing$ and moreover, if $\sigma \inseg \tau$ are such that $W_p^{|\sigma|}\cap\ng(\sigma)=W_p^{|\tau|}\cap\ng(\tau)$, then $\Xi(\sigma,p)=\Xi(\tau,p)$.
One way to do this is by letting $\Xi(\sigma,p)$ denote the unique program, which given $x$ successively checks, whether $x=y_i$, where 
$(y_i)_{i<|\ng(\sigma)|}$ is the increasing enumeration of $\ng(\sigma)$. As soon as the answer is positive, the program goes into a loop. Otherwise it executes the program encoded in $p$ on $x$, which yields
$$
\varphi_{\Xi(\sigma,p)}(x)=
\begin{cases}
\uparrow, &\text{if } x\in\ng(\sigma);\\
\varphi_p(x), &\text{otherwise.}
\end{cases}
$$

\smallskip
Now, $M'$ works as follows:
\begin{enumerate}[label=\Roman*.]
\item Compute $p_i:=M(\sigma[i])$ for all $i\leq|\sigma|$.
\item Withdraw all $p_i$ with the property
		$|\ng(\sigma)\cap W_{p_i}^{|\sigma|}|> a$.
\item Define $M'(\sigma)$ to be a code for the program coresponding to the union vote of all $\Xi(\sigma,p_i)$, for which $p_i$ was not withdrawn in the previous step:

\begin{quote}
Given input $x$, for $n$ from 0 till $\infty$ do the following:
If $i := \pi_1(n)\leq|\sigma|$, $|\ng(\sigma)\cap W_{p_i}^{|\sigma|}|\leq a$ and $\Phi_{\Xi(\sigma,p_i)}(x)\leq \pi_2(n)$, then return 0; otherwise increment $n$.
\end{quote}

This guarantees
$$\varphi_{M'(\sigma)}(x)=\begin{cases}
0, &\text{if } \exists \,i \leq |\sigma| \:(\, |\ng(\sigma)\cap W_{p_i}^{|\sigma|}|\leq a \:\wedge\: \varphi_{\Xi(\sigma,p_i)}(x)\!\!\downarrow \,); \\
\uparrow, &\text{otherwise.}
\end{cases}$$
\end{enumerate}
Intuitively, $M'(\sigma)$ eliminates all membership errors in guesses of $M$ on initial segments of $\sigma$, not immediately violating the allowed number of anomalies, and then asks whether one of them converges on the input,
which implies $$W_{M'(\sigma)}=\bigcup_{i\leq|\sigma|, |\ng(\sigma)\cap W_{p_i}^{|\sigma|}|\leq a} W_{\Xi(\sigma,M(\sigma[i]))}.$$

\smallskip
In order to show $\CalL \in \Inf\Lim^a(M')$, let $L \in \CalL$ and $I \in \Inf(L)$.
As $\CalL \in \Lim^a_\ast(M)$, there is $t_0$ such that all of $M$'s hypotheses on $I$ are in $\{h_s\mid s\leq t_0\}$ and additionally $|\, W^{t_0}_{h_s}\cap \N\setminus L \,|>a$ for all $s \leq t_0$ with $|\,W_{h_s} \cap \N\setminus L| > a$.
Moreover, we can assume that for all $s\leq t_0$ with $|\, W_{h_s}\cap \N\setminus L \,|\leq a$ we have observed all commission errors in at most $t_0$ steps, which formally reads as $W_{h_s}\cap \N\setminus L = W_{h_s}^{t_0}\cap \N\setminus L$.

Then for all $t \geq t_0$ we obtain the same set of indices $$A:=\{\,\Xi(I[t],p_i) \mid i\leq t \wedge |\ng(I[t])\cap W^t_{p_i}|\leq a\,\}$$ and therefore $M'$ will return syntactically the same hypothesis, namely, $h'_{t_0}$.

It remains to argue for $W_{h'_{t_0}}=^a L$. By construction and the choice of $t_0$ there are no commission errors, i.e., $W_{h'_{t_0}} \cap \N\setminus L =\varnothing$. Further, since $\varphi_{h'_{t_0}}(x)$ exists in case there is at least one $p\in A$ such that $\varphi_{p}(x)$ exists, there are at most $a$ arguments, on which $\varphi_{h'_{t_0}}$ is undefined.
\end{proof}

This contrasts the results in language learning from text in \cite{case1999power}, observing for every $a \in \N\cup\{\ast\}$ a hierarchy
\begin{align*}
[\Txt\Lim^a] &\subsetneq \ldots \subsetneq [\Txt\Lim^a_b] \subsetneq [\Txt\Lim^a_{b+1}] \subsetneq \ldots \\
&\subsetneq \bigcup_{b \in \N_{>0}} [\Txt\Lim^a_b] \subsetneq  [\Txt\Lim^a_\ast] \subseteq [\Txt\Bc^a].
\end{align*}

For learning from informant we gain for every $a \in \N\cup\{\ast\}$ a duality
\begin{align*}
[\Inf\Lim^a] & = \ldots = [\Inf\Lim^a_b] = [\Inf\Lim^a_{b+1}] = \ldots \\
&= \bigcup_{b \in \N_{>0}} [\Inf\Lim^a_b] =  [\Inf\Lim^a_\ast] \subsetneq [\Inf\Bc^a].
\end{align*}

\section{Learning Characteristic Functions of Collections of Recursive Languages}
\label{sec:RecInd-RecLan}
We now turn to the setting in which we want to learn a set of Boolean classifiers.
In Machine Learning the input is usually considered a labeled element of $\mathbb{R}^d$.
It is reasonable to consider only the countably many $d$-tuples $x$ of computable reals $\mathbb{R}_\mathrm{comp}^d$.
By fixing a (non-computable) enumeration $\mathbb{R}_\mathrm{comp}^d=\langle x_i \mid i<\N \rangle$, we might as a first attempt identify $i$ with $x_i$.
Then our hypothesis space is the set of all Boolean functions.
We will later restrict ourselves to total computable Boolean functions.

\bigskip
Definitions \ref{def:informant} for informant and \ref{def:learner} for the learner are independent of the interpretation of the hypothesis.
The Definition~\ref{def:InfConvergence} of convergence criteria has to be slightly modified as follows.

\begin{definition}
\label{def:RecInd-convergence}
Let $M$ be a learner and $\CalL$ a collection of recursive languages. Further, let $a \in \N\cup\{\ast\}$ and $b \in \N_{>0}\cup\{\ast,\infty\}$.
\begin{enumerate}
\item Let $L \in \CalL$ be a language and $I \in \Inf(L)$ an informant for $L$ presented to $M$.
\begin{enumerate}
\item We call $h=(h_t)_{t\in\N} \in \ISeq{(\N\cup\{?\})}$, where $h_t := M(I[t])$ for all $t \in \N$, the \emph{learning sequence of $M$ on $I$}.
\item $M$ \emph{learns $L$ from $I$ with $a$ anomalies and vacillation number $b$ in the limit}, for short $M$ ${\Lim_C}^a_b$-learns $L$ from $I$ or ${\Lim_C}_b^a(M,I)$, if there is a time $t_0 \in \N$ such that $|\,\{\, h_t \mid t \geq t_0 \,\}\,| \leq b$ and for all $t \geq t_0$ we have
$\Diff_L(h_t) = \{ x\in\N \mid \varphi_{h_t}(x)\neq\classifier_L(x) \}$ has at most size $a$.
\end{enumerate}
\item $M$ \emph{learns $\CalL$ with $a$ anomalies and vacillation number $b$ in the limit}, for short $M$ ${\Lim_C}^a_b$-learns $\CalL$, if ${\Lim_C}^a_b(M,I)$ for every $L\in\CalL$ and every $I \in \Inf(L)$.
\end{enumerate}
\end{definition}

This is also equivalent to learning the characteristic function of $L$ from text.

We also have to adjust the Definition~\ref{def:lockingsequence} of locking sequences.

\begin{definition}
\label{def:RecInd-lockingsequence}
Let $M$ be a learner, $L$ a language and $a\in\N\cup\{\ast\}$ as well as $b\in\N_{>0}\cup\{\ast,\infty\}$. We call $\sigma \in \Seq{(\Nttwo)}$ an \emph{${\Lim_C}^a_b$-locking sequence for $M$ on $L$}, if $\Comp(\sigma,L)$ and
\begin{align*}
\exists D \subseteq \N \:\Big(\, &|D|\leq b \:\wedge\: \forall \tau \in \Seq{(\Nttwo)} \Big(\, \Comp(\tau,L) \Rightarrow \\
&\left.\left. \left( M(\sigma\concat\tau)\!\!\downarrow \wedge \,|\Diff_L(M(\sigma\concat\tau))|\leq a \wedge M(\sigma\concat\tau)\in D \,\right)\right) \,\right)
\end{align*}
\end{definition}

Then the proof of Lemma~\ref{lem:lockingsequence} immediately transfers and we obtain the following lemma.

\begin{lemma}
\label{lemma:RecInd-lockingsequence}
Let $M$ be a learner, $a\in\N\cup\{\ast\}$, $b\in\N_{>0}\cup\{\ast,\infty\}$ and $L$ a language ${\Lim_C}^a_b$-identified by $M$. Then there is a ${\Lim_C}^a_b$-locking sequence for $M$ on $L$.
\end{lemma}

We also have to adjust the Definition~\ref{def:consistent} of consistency in the following way.

\begin{definition}
\label{def:RecInd-consistent}
Let $\varphi$ be a Boolean computable function.
We define
\begin{align*}
\ps(\varphi) = \{ x\in\N\mid \varphi(x)\!\!\downarrow\,=1\}; \\
\ng(\varphi) = \{ x\in\N\mid \varphi(x)\!\!\downarrow\,=0\}.
\end{align*}
Let $f \in \IfiSeq{(\Nttwo)}$. We say \emph{$f$ is consistent with $\varphi$}, for short $\Comp(f,\varphi)$, if
\begin{align*}
\ps(f) \subseteq \ps(\varphi) \;\wedge\; \ng(f) \subseteq \ng(\varphi).
\end{align*}
\end{definition}

Let $C_{h_t}$ denote $\ps(\varphi_{h_t})$.
By replacing $W_{h_t}$ by $C_{h_t}$, the definitions of the learning restrictions in Definition~\ref{def:learningrestrictions}, learning success criteria in Definition~\ref{def:learningsuccesscriterion} and learning criteria in Definition~\ref{def:learningcriterium} remain the same.
The implications (independent of the learning success criterion at hand) between the delayable learning restrictions as stated in Lemma~\ref{lem:delayablebackbone} hold accordingly.

\medskip
Moreover, the Definition~\ref{def:delayable} and basic Lemma~\ref{lem:delayable} concerning delayability remain unchanged.
Also Lemma~\ref{lem:canInfSdDelayableEx} about considering canonical informant being sufficient and Lemma~\ref{lem:TotalInfDelayableEx} about totality being no restriction for delayable learning success criteria still hold as the proofs only refer to the concept of delayability.

\bigskip
To our knowledge Machine Learning algorithms only hypothesize total classifiers.
Denote the set of encoded programs for total Boolean functions on $\N$ by $\CInd$.
From now on we will allow the learner $M$ to hypothesize elements of $\CInd$ on data consistent with some classifier to be learned only.
We denote by $[\Inf\CInd\Lim_C]$ the collection of all recursive languages $\Lim_C$-learnable by such a learner $M$ from informant.
In Definition~\ref{def:learningsuccesscriterion} in the learning success criterion at position $\beta$, we write $\CInd$ between the learning restrictions to be met and the convergence criterion.

With $[\CInd\Inf\Lim_C]$ we refer to the collection of all recursive languages $\Lim_C$-learnable by a learner with range contained in $\CInd$.
These learners only output hypotheses for total computable Boolean functions and in Definition~\ref{def:learningsuccesscriterion} we write $\CInd$ as part of $\alpha$.

Later we might consider appropriately chosen subsets of $\CInd$ as hypothesis space.

\medskip
In this setting we can assume the learner to output only hypotheses consistent with the input on relevant data.
This is done by patching the hypothesis according to the finitely many training data points the learner has received so far.

\begin{proposition}
\label{prop:RecInd-canInfConsEx}
We have
\begin{equation*}
        [\Inf_\can\CInd\Ex]=[\Sd\Inf\Cons\CInd\Ex].
\end{equation*}
\end{proposition}
\begin{proof}
We use the idea from Lemma~\ref{lem:canInfSdDelayableEx}.
Thus, the new learner outputs $M$'s hypothesis $h$ on the largest complete canonical informant with information only from the current input $\sigma$.
As $h$ is an index for a total function, we can, in a uniformly computable way, obtain a hypothesis $h_\sigma$ from $h$ such that
\begin{enumerate}
\item $\varphi_{h_\sigma}$ is consistent with all data in $\sigma$ and
\item $h_\sigma=h$ if $\sigma$ is consistent with $\varphi_h$.
\end{enumerate}
More precisely, the computable operator maps an index $h$ of a computable function $\varphi_h: \N \to \{0,1\}$ and a finite informant sequence $\sigma$ to an index $h_\sigma$ of a computable function $\varphi_{h_\sigma}$ with
\begin{align*}
\varphi_{h_\sigma}(x)=\begin{cases}
1, &\text{if } x \in \ps(\sigma);\\
0, &\text{else if } x \in \neg(\sigma);\\
\varphi_h(x), &\text{otherwise.}
\end{cases}
\end{align*}
The simulation only requires information about $\ps(\sigma)\cup\ng(\sigma)$ and thus the learner is set-driven.
Further, $h_{\sigma}=h$ whenever $\varphi_h$ is consistent with $\sigma$.
As $M$ converges on the canonical informant and we only alter $h$ in case at least one datum in $\sigma$ is inconsistent with $\varphi_h$, we obtain the convergence of the new learner.
Clearly, it is consistent by construction.
\end{proof}

Summing up, as consistency of the input data with a hypothesized total computable Boolean functions is decidable, $\CInd$-learners can be assumed consistent while learning.
By the same argument $\tau(\CInd)$-learners can be assumed $\tau(\Cons)$.

\bigskip
It is easy to see that $\Ex$ can be replaced by every convergence criterion (and also $\Mon$, $\NU$, $\SNU$).

\bigskip
On the other hand, it is easy to adapt the proof of Proposition~\ref{prop:TotalInfConsEx} as follows.

\begin{proposition}
\label{prop:RecInd-TotInfConsEx}
There is a \emph{collection of decidable languages} witnessing
$$[\totalCp\Inf\Cons\CInd{\Lim_C}] \subsetneq [\Inf\Cons\CInd{\Lim_C}].$$
\end{proposition}
\begin{proof}
Let $o$ be an index for the everywhere $0$-function.
Further, define for all $\sigma \in \Seq{(\Nttwo)}$ the learner $M$ by
$$
M(\sigma) := \begin{cases}
o, &\text{if } \ps(\sigma)=\varnothing; \\
\varphi_{\max(\ps(\sigma))}(\langle\sigma\rangle), &\text{otherwise.}
\end{cases}
$$
We argue that $\CalL := \{\, L \subseteq \N \mid L \text{ is decidable and } L\in\Inf\Cons{\Lim_C}(M) \,\}$ is not consistently learnable by a total learner from informant.
Assume towards a contradiction $M'$ is such a learner.
For a sequence $\sigma$ of natural numbers we denote by $\overline{\sigma}$ the corresponding canonical finite informant sequence, ending with the highest value $\sigma$ takes. Further, for a natural number $x$ we denote by $\tau(x)$ the unique element of $\Seq{\N}$ with $\langle\tau(x)\rangle=x$.
Then by padded ORT there are $e, z \in \N$ and functions $a, b: \Seq{\N} \to \N$, such that
\begin{equation}
\label{eq:CInd-abincreasing}
\forall \sigma, \tau \in \Seq{\N} \: (\, \sigma \prinseg \tau \Rightarrow \max\{a(\sigma),b(\sigma)\} < \min\{a(\tau),b(\tau)\}\,),
\end{equation}
with the property that for all $\sigma \in \Seq{\N}$ and all $i \in \N$
\begin{align}
\sigma_0 &= \varnothing; \nonumber \\
\sigma_{i+1} &= \sigma_i\,\concat \begin{cases}
a(\sigma_i), &\text{if } M'(\overline{\sigma_i\:\!\concat a(\sigma_i)})\neq M'(\overline{\sigma_i}); \\
b(\sigma_i), &\text{otherwise;}
\end{cases} \label{eq:CInd-MCNewLearner} \\
\varphi_e(y) &= \begin{cases}
1, &\text{if } y \in \ps(\overline{\sigma_y});\\
0, &\text{otherwise;}
\end{cases} \nonumber \\
\varphi_{a(\sigma)}(x)&=\begin{cases}
e, & \text{if } \Comp(\tau(x),\varphi_e) \text{ and } M'(\overline{\sigma\concat a(\sigma)})\neq  M'(\overline{\sigma}); \\
\ind(\ps(\tau(x))), & \text{otherwise;}
\end{cases} \nonumber \\
\varphi_{b(\sigma)}(x) &= \begin{cases}
e, & \text{if } \Comp(\tau(x),\varphi_e); \\
\ind(\ps(\tau(x))), & \text{otherwise;}
\end{cases} \nonumber
\end{align}
Consider the decidable language $L_e=\ps(\varphi_e)$.
Clearly, we have $L_e \in \CalL$ and thus $M'$ also $\Inf\Cons{\Lim_C}$-learns $L_e$.
By the ${\Lim_C}$-convergence there are $e', j \in \N$, where $j$ is minimal, such that $\varphi_{e'}=\varphi_e$ and for all $i \geq j$ we have $M'(\overline{\sigma_i})=e'$ and hence $M'(\overline{\sigma_i\:\!\concat a(\sigma_i)})= M'(\overline{\sigma_i})$ by \eqref{eq:CInd-MCNewLearner}.

We now argue that $L := \ps(\overline{\sigma_j}) \cup \{a(\sigma_j)\} \in \CalL$.
Let $I$ be an informant for $L$ and $t \in \N$.
By \eqref{eq:CInd-MCNewLearner} we observe that $M$ is consistent on $I$ as
$$M(I[t])=\varphi_{\max(\ps(I[t]))}(\langle I[t] \rangle) = \begin{cases}
e, &\text{if } \Comp(I[t],\varphi_e); \\
\ind(\ps(I[t])), & \text{otherwise.}
\end{cases}$$
Further, by the choice of $j$ we have $\neg\Cons(\,(a(\sigma_j),1),\varphi_e\,)$.
If $\ps(I[t])=L$, we obtain $\varphi_{M(I[t])}=\ind_L$.
On the other hand $M'$ does not consistently learn $L$ as by the choice of $j$ we obtain $M'(\overline{\sigma_j\!\,\concat a(\sigma_j)})=M'(\overline{\sigma_j})=e'$ and $\neg \Comp(\overline{\sigma_j\!\,\concat a(\sigma_j)},L_e)$, a contradiction.
\end{proof}

Thus, learning algorithms not defined on all inputs have strictly more learning power. 

\bigskip
As we clearly can do a padding-trick for $C$-indices, similar to Lemma~\ref{lem:InfSynDecEx}, we might assume the learner to be syntactically decisive.
Furthermore, the separations of $\Caut$, $\Mon$ and $\SMon$ are still valid as they are witnessed by indexable families.
Thus, the interesting question is whether $\Conv$ and $\SDec$ are also not restrictive for binary classifiers.
We now observe that this still holds true but the proof is much simpler than for $W$-indices, because the consistency of data with hypotheses is decidable.

\begin{theorem}
\label{thm:RecInd-ConvSDecInfEx}
For $\delta \in \{\mathbf{T},\Mon\}$ holds
$$[\Inf\delta\CInd\Bc_C]=[\Inf\Conv\SDec\delta\CInd\Ex_C].$$
\end{theorem}
\begin{proof}
By the comment after Proposition~\ref{prop:RecInd-canInfConsEx} we assume $\delta\subseteq\Cons$.
Let $\CalL\in[\Inf\delta\CInd\Bc_C]$ and the learner $M$ witnessing this.
It is an easy exercise to check that the following learner acts as required, where $\sigma$ is a finite informant sequence and $\xi\in\N\times\{0,1\}$:
\begin{align*}
M'(\varnothing)&=M(\varnothing);\\
M'(\sigma\concat\xi)&=
\begin{cases}
M(\sigma\concat\xi), &\text{if } \neg\Comp(\sigma\concat\xi,M'(\sigma));\\
M'(\sigma), &\text{otherwise.}
\end{cases}
\end{align*}
Note that the consistency of $M$ on $\CalL$ is only employed to obtain $\SDec$.
\end{proof}

\begin{corollary}
$[\Inf_\can\CInd\Bc_C]=[\Inf\Cons\Conv\CInd\Ex_C].$
\end{corollary}

For $\tau(\CInd)$-learners the simulation in Theorem~\ref{thm:RecInd-ConvSDecInfEx} preserves totality.

\bigskip
In a nutshell for learners only outputting $C$-indices, we obtain the same map as for $W$-indices.
In contrast, $\Cons$ is not a restriction anymore.

\bigskip
Moreover, $\Bc_C$-learning is not weaker than explanatory learning and thus the vacillatory hierarchy collapses.

\section{Further Research}
Future investigations could address the relationships between the different delayable learning restrictions for other convergence criteria, where the general results in Section \ref{sec:WlogTotalMethodical} may be helpful.

\medskip
Another open question regards the relation between learning recursive functions from text for their graphs and learning languages from either informant or text. It seems like delayability plays a crucial role in order to obtain normal forms and investigate how learning restrictions relate in each setting. 
It is yet not clear, whether delayability is the right assumption to generalize Lemma \ref{FnResultsToInf}.
The survey \cite{Zeu-Zil:j:08:Survey} and the standard textbook \cite{Jai-Osh-Roy-Sha:b:99:stl2} contain more results in the setting of function learning which may transfer to learning collections of languages from informant with such a generalization.


\medskip
According to \cite{STL1} requiring the learner to base its hypothesis only on the previous one and the current datum, makes $\Lim$-learning harder.
While the relations between the delayable learning restrictions for these so called \emph{iterative learners} in the presentation mode of solely positive information has been investigated in \cite{jain2016role}, so far this has not been done when learning from informant.
For indexable families, this was already of interest to \cite{lz-tmllc-92}, \cite{lange2003variants} and \cite{Jai-Lan-Zil:j:07}.
In the corresponding map each of $\Caut$, $\Mon$ and $\SMon$ is separated from all other learning restrictions.
Moreover, $\Conv$ restricts iterative learning from informant and we are sure that also $\SNU$ does.
It remains open, whether all syntactic learning criteria have the same learning power.
Further, it seems like settling $\NU$, $\Dec$ and $\WMon$ requires completely new techniques.
This model is of special interest as it models the behavior of neural networks.
Further improvements to the model would be a more appropriate hypothesis space, a probabilistic presentation of the data and other convergence criteria.
For $C$-indices the incomparability of $\Caut$ and $\Mon$, as well as the separation of $\Conv$ are still valid.

\medskip
For automatic structures as alternative approach to model a learner, there have been investigations on how different types of text effect the $\Lim$-learnabi\-lity, see \cite{jain2010learnability} and \cite{pmlr-v76-holzl17a}. The latter started investigating how learning from canonical informant and learning from text relate to one another in the automatic setting. A natural question seems to be what effect other kinds of informant and learning success criteria have.

\medskip
Last but not least, rating the models value for other research aiming at understanding the capability of human and machine learning is another very challenging task to tackle.

\subsection*{Acknowledgments and Funding}
This work was supported by the German Research Foundation (DFG) under Grant KO 4635/1-1 (SCL).

\bibliographystyle{alpha}
\bibliography{CLTBib}

\end{document}